\documentclass{m2an}
\usepackage{amsmath, amssymb, amsfonts}
\usepackage{mathrsfs}
\usepackage{graphicx,color}
\usepackage{booktabs}
\usepackage{diagbox}
\usepackage{algorithm}
\usepackage{algpseudocode}
\usepackage{enumitem}
\usepackage{multirow}
\usepackage{subfig}
\usepackage{epstopdf}
\usepackage{comment}
\usepackage{adjustbox}
\usepackage{braket}
\usepackage{bm}
\usepackage{bbm}
\usepackage{hyperref}

\newcommand{\Tr}{\mathrm{Tr}}

\newcommand{\ud}{\,\mathrm{d}}

\newcommand{\Or}{\mathcal{O}}

\newcommand{\I}{\mathrm{i}}

\newcommand{\RR}{\mathbb{R}}
\newcommand{\CC}{\mathbb{C}}
\newcommand{\ZZ}{\mathbb{Z}}

\global\long\def\Tr{\mathrm{Tr}}

\numberwithin{equation}{section}
\numberwithin{figure}{section}
\newtheorem{thm}{\protect\theoremname}

\newtheorem{prop}[thm]{\protect\propositionname}

\providecommand{\corollaryname}{Corollary}
\providecommand{\lemmaname}{Lemma}
\providecommand{\propositionname}{Proposition}

\providecommand{\theoremname}{Theorem}

\newcommand{\CS}{\mathcal{C}}

\begin{document}

\title{Numerical solution of large scale Hartree-Fock-Bogoliubov equations}
\author{Lin Lin}\address{Department of Mathematics, University of California, Berkeley, and Computational Research Division, Lawrence Berkeley National Laboratory, Berkeley, CA 94720. Email: \texttt{linlin@math.berkeley.edu}}
\author{Xiaojie Wu}\address{Department of Mathematics, University of California,  Berkeley, CA 94720. Email: \texttt{xiaojiewu@berkeley.edu}}

\begin{abstract}
The Hartree-Fock-Bogoliubov (HFB) theory is the starting point for treating superconducting systems. However, the computational cost for solving large scale HFB equations can be much larger than that of the Hartree-Fock equations, particularly when the Hamiltonian matrix is sparse, and the number of electrons $N$ is relatively small compared to the matrix size $N_{b}$. We first provide a concise and relatively self-contained review of the HFB theory for general finite sized quantum systems, with special focus on the treatment of spin symmetries from a linear algebra perspective. We then demonstrate that the pole expansion and selected inversion (PEXSI) method can be particularly well suited for solving large scale HFB equations. For a Hubbard-type Hamiltonian, the cost of PEXSI is at most $\Or(N_b^2)$ for both gapped and gapless systems, which can be significantly faster than the standard cubic scaling diagonalization methods. We show that PEXSI can solve a two-dimensional Hubbard-Hofstadter model with $N_b$ up to $2.88\times 10^6$, and the wall clock time is less than $100$ s using $17280$ CPU cores. This enables the simulation of physical systems under experimentally realizable magnetic fields, which cannot be otherwise simulated with smaller systems. 
\end{abstract}

\maketitle

\section{Introduction}

The Hartree-Fock (HF) theory plays a fundamental role in quantum physics and chemistry. Similarly, the  Hartree-Fock-Bogoliubov (HFB) theory is the simplest first principle method for treating superconducting systems. The HFB theory generalizes the celebrated Bardeen-Cooper-Schrieffer (BCS) theory~\cite{bcs0,bcs1}, which successfully explained  superconducting phenomena when the phase transition temperature (denoted by $T_c$) is low. Even for high $T_c$ superconductors (HTC) where the HFB theory is not effective by itself, it is still the starting point and a key component in more advanced theories, such as those based on quantum Monte Carlo methods~\cite{chiesa2013,Rosenberg2015,shiwei2017} and quantum embedding theories~\cite{Zheng2016,Zheng2017}. 
A distinctive feature of  HFB  is the existence of the \textit{pairing effect}. When pairing occurs,  the number of particles is no longer a good quantum number, and a fluctuating number of particles is the premise of the superconducting phenomena.   

Consider a quantum many-body system discretized using $N_b$ basis functions (called spin orbitals in quantum physics literature) with $N$ electrons (for HFB, this is defined in the sense of an ensemble average).  Both the HF and HFB equations are nonlinear equations, which can be solved iteratively via the self-consistent field (SCF) iteration. However, the computational cost for solving the HFB equations can be much higher, due to the following two reasons. 1) HF calculations require computing the lowest $N$ eigenfunctions of a Hamiltonian matrix of size $N_b$, while HFB requires computing the lowest $N_b$ eigenfunctions of a quasi-Hamiltonian matrix of size $2N_b$, i.e. $50\%$ of the eigenpairs.  This essentially forbids the efficient usage of iterative eigensolvers, even if $N_b \gg N$ (such as the case of a large basis set or low-doping) or if the HFB Hamiltonian is sparse. 2) The constraint of the number of electrons in the HF theory can be trivially satisfied by taking the lowest $N$ eigenfunctions. Such a constraint in the HFB theory can only be satisfied by iteratively adjusting the chemical potential, which can increase the number of SCF iterations. Therefore in practice it can be difficult to perform HFB calculations for large scale systems, such as the treatment of superconductors under an experimentally realizable magnetic field~\cite{zhu2016bogoliubov}, large quantum dots~\cite{goldhaber1998kondo, cronenwett1998tunable}, nano-transport phenomena~\cite{Kosov2009}, to name a few.

\noindent\textbf{Contribution:}

The contribution of this paper is two-fold. First, we attempt to provide
a relatively concise and self contained derivation of the HFB theory.
While the HFB theory itself is certainly well-known in the physics literature, our experience indicated that the derivation of its general form for finite sized quantum systems, and the associated linear algebra structures, are not commonly presented in detail. This is partly because many  textbooks   in physics discuss the HFB and BCS theories together, and often focus on certain special cases such as the spin-singlet coupling case or translational invariant
systems. While these settings indeed occur most frequently in practice, if one starts from such settings, some amount of reverse engineering could be needed to grasp the overall picture.  Our perspective largely follows that of the excellent textbook by Blaizot and Ripka~\cite{Blaizot1986}\footnote{Unfortunately, we found that this book seems to be out of production now.}, which focuses on finite sized systems and therefore can be described using finite dimensional matrices. We compare the difference between the numerical solution of HF equations and HFB equations, especially in the case when $N_b$ is large. We also introduce the concept of spin-singlet and spin-triplet couplings from a linear algebra perspective, which reduces the dimension of the Hamiltonian. We hope our presentation would be useful for interested readers not familiar with the matter. 


Second, we propose to accelerate large scale HFB calculations using the
pole expansion and selected inversion (PEXSI) method~\cite{LinLuYingE2009,LinGarciaHuhsEtAl2014}, which is a Fermi operator expansion (FOE) method for solving electronic structure problems \cite{Goedecker1999}. 
While the computational cost for diagonalizing the HFB Hamiltonian scales as $O(N_b^3),$  PEXSI can evaluate the generalized density matrix with cost that scales at most as $\Or(N_b^2)$ for certain sparse Hamiltonians (such as Hubbard-type Hamiltonians). As mentioned before, diagonalization methods for HF calculations only need to evaluate the lowest $N$ eigenpairs, and can immediately identify the chemical potential once the eigenvalues are available. On the other hand, the cost of PEXSI only depends on $N_b$ and is independent of $N$, and the chemical potential can only be determined iteratively. Therefore for HF calculations, PEXSI only becomes faster than diagonalization methods when the system size becomes relatively large. The advantage of diagonalization methods no longer holds for HFB calculations, and therefore PEXSI can become advantageous at  rather small system sizes. Furthermore, the PEXSI method is ideally suited for parallel computing, and can be scaled to $10^4\sim 10^5$ processors.   

Using a two-dimensional Hubbard-Hofstadter model for
example, we demonstrate that the PEXSI method can already be more
efficient than diagonalization methods for small systems of less than $400$ sites. 
Thanks to the reduced complexity, we perform large scale HFB
calculations for systems of $1.44\times 10^6$ sites, and the wall clock time of each calculation of grand canonical ensemble is less than $100$ seconds using $17280$ CPU cores. A
diagonalization method would require using a dense eigensolver for a
complex matrix of size $2.88\times 10^6$, which is prohibitively expensive.
This allows us to approach the experimentally realizable range of magnetic fields of around $20$ Tesla, using a lattice of  around $2\times 10^6$  sites (and the matrix size is $4\times 10^6$). We also demonstrate the
usage of the PEXSI based HFB method for studying phase diagrams and striped order of the pairing potential for large systems.

\noindent\textbf{Related works:}
The general mathematical formulation of of the HFB theory was studied comprehensively by Bach, Lieb and Solovej~\cite{Bach1994}, which leads to many subsequent works \cite{kuzemsky2015, braunlich2014, Hainzl2005, Hainzl2007, lenzmann2010}. The first study HFB from the perspective of numerical analysis only appeared a few years
ago by Lewin and Paul~\cite{LewinPaul}, which focused on the self-consistent field iterations in
HFB calculations. Besides diagonalization methods, Fermi operator expansion methods (based on the Chebyshev expansion) have also been used to accelerate the solution of HFB equations \cite{Nagai2012,covaci2010,zha2010}. 

\noindent\textbf{Organization:}
The rest of the paper is organized as follows. We first introduce some background information, including the HF theory and a corollary of Wick's theorem in Section~\ref{sec:hf}. This allows us to introduce the
HFB theory in Section~\ref{sec:hfb}, in a way that is parallel to the discussion of the HF theory. The numerical solution of HFB equations using diagonalization methods, and the symmetry considerations are discussed in Section \ref{sec:solvehfb}. We then
introduce the PEXSI method for HFB calculations in
Section~\ref{sec:pexsi}. We demonstrate the numerical
performance of PEXSI in Section~\ref{sec:numer}. The conclusion and discussion are given in
Section~\ref{sec:conclusion}. 

\section{Preliminaries}\label{sec:hf}

%
%
%
%
%
%
%

Throughout the paper, we use the Dirac bra-ket notation
for quantum states. For a matrix $A\in\CC^{m\times n}$, its transpose,
complex conjugate (entry-wise), and Hermitian conjugate are
denoted by $A^{\top},\overline{A},A^{\dagger}$, respectively.  Unless otherwise mentioned, a vector is always viewed as a column vector.

In the second-quantized formulation, the state space is called the Fock space, denoted by $\mathcal{F}$. The Fock space is the direct sum of tensor products of multiple replicas of single-particle Hilbert space $\mathcal H$. Given a basis $\{\ket{\psi_i}\}_{i=1,\ldots,N_b}$ of $\mathcal H$, the occupancy number basis set for the Fock space is 
\[
  \{\ket{s_1,\ldots,s_{N_b}}=\ket{\psi_1}^{s_1}\ldots\ket{\psi_{N_b}}^{s_{N_b}}\}, \quad s_i \in \{0,1\},i=1,\ldots,N_b,
\]
which is an orthonormal basis set satisfying
\begin{equation}
\braket{s_{i_1},\ldots,s_{i_{N_b}} \vert s_{j_1},\ldots,s_{j_{N_b}}} =
\delta_{i_1 j_1}\cdots \delta_{i_{N_b} j_{N_b}}.
\end{equation}
Here $N_b$ is the number of basis functions, or the number of sites in the single-particle space $\mathcal H$. Hence, the dimension of Fock space is $2^{N_b}$.

A state $\ket{\Psi}\in\mathcal{F}$ will be written as a
linear combination of occupancy number basis elements as follows: 
\begin{equation}
  \ket{\Psi} = \sum_{s_1,\ldots,s_{N_b}\in\{0,1\}}
  \Psi(s_1,\ldots,s_{N_b}) \ket{s_1,\ldots,s_{N_b}}, \quad
  \Psi(s_1,\ldots,s_{N_b})\in\CC.
  \label{}
\end{equation}
Hence the state vector $\ket{\Psi}$ can be identified with a vector
$\Psi\in\CC^{2^{{N_b}}}$, and $\mathcal{F}\cong\CC^{2^{{N_b}}}$.
Without loss of generality we always assume $\ket{\Psi}$ is normalized,
i.e.
\begin{equation}
  \braket{\Psi|\Psi} = \sum_{s_1,\ldots,s_{N_b}\in\{0,1\}}
  |\Psi(s_1,\ldots,s_{N_b})|^2 = 1.
  \label{eqn:psi_normalize}
\end{equation}

The fermionic creation and annihilation operators, which add and remove one particle in the quantum state, are respectively defined as $(p= 1,\ldots,N_b)$
\begin{equation}
  \begin{split}
    \hat{a}_p^\dagger \ket{s_1,\ldots,s_{N_b}} &= (-1)^{\sum_{q=1}^{p-1}
      s_{q}}(1-s_p)\ket{s_1,\ldots,1-s_{p},\ldots,s_{N_b}},\\
    \hat{a}_p \ket{s_1,\ldots,s_{N_b}} &= (-1)^{\sum_{q=1}^{p-1} s_{q}}s_p\ket{s_1,\ldots,1-s_{p},\ldots,s_{N_b}}.
  \end{split}
\end{equation}
They satisfy the canonical anti-communication relation (CAR)
\begin{equation}
  \{\hat{a}_p^{\dagger},\hat{a}_q\}=\delta_{pq},\quad \{\hat{a}_p^{\dagger},\hat{a}_q^{\dagger}\}=0,\quad 
  \{\hat{a}_p,\hat{a}_q\}=0, \quad p,q=1,\ldots,N_b.
  \label{eqn:anti_commutation}
\end{equation}
The number operator defined as $\hat{n}_{p}:=\hat{a}^{\dagger}_{p} \hat{a}_{p}$ satisfies
\begin{equation}
  \hat n_p \ket{s_1,\ldots,s_{N_b}} = s_p \ket{s_1,\ldots,s_{N_b}},\quad p= 1,\ldots,N_b.
\end{equation}
The eigenvalues of number operators are either $0$ or $1$. We also define $\hat{N}=\sum_{p=1}^{N_b} \hat{n}_p$ to be the total
number operator.

The state $\ket{0}:=\ket{0,\ldots,0}$ is called the vacuum state. 
In particular, an annihilation operator acting on the vacuum state always
vanishes, i.e. 
\begin{equation}
  \hat{a}_{p}\ket{0}=0, \quad p=1,\ldots,N_b.
  \label{eqn:vacuum}
\end{equation}
In fact, Eq.~\eqref{eqn:vacuum} can also be viewed as the \textit{defining equation}
for the vacuum corresponding to a set of creation and annihilation operators satisfying the CAR. 
For a given state $\ket{\Psi}\in\mathcal{F}$ and a
self-adjoint operator $\hat{O}$, we
define $\braket{\hat{O}}:=\braket{\Psi|\hat{O}|\Psi}$
to be the expectation value of $\hat{O}$.

In this paper, we assume that the quantum many-body Hamiltonian takes the form
\begin{equation}\label{eqn:manybodyH}
\hat{H}=\sum_{p,q=1}^{N_b} h_{pq}^{0} \hat{a}_{p}^{\dagger} \hat{a}_{q}+\frac{1}{4} \sum_{p,q,r,s=1}^{N_b} V_{p q r s} \hat{a}_{p}^{\dagger} \hat{a}_{q}^{\dagger} \hat{a}_{s} \hat{a}_{r}
\end{equation}
Here $h^{0}\in \CC^{N_b\times N_b}$ is a Hermitian matrix. 
The superscript $0$ means that $h^0$ only comes from the single-particle contribution. 
$V$ is a 4-tensor and characterizes the two-particle interaction. Due to the CAR, without loss of generality we may require $V$ to satisfy the following symmetry properties
\begin{equation}\label{eqn:Vsymmmetry}
  V_{pqrs}=-V_{pqsr}=-V_{qprs}=V_{qpsr}.
\end{equation}
In other words, $V\in \CC^{N_b\times N_b\times N_b\times N_b}$ is
an anti-symmetric $4$-tensor with respect to the permutation of indices $p,q$ or
$r,s$.  Furthermore, $V_{pqrs}=\overline{V}_{rspq}$.
More general Hamiltonians such as those containing a cubic term with respect to the creation and annihilation operators can be treated similarly. Note that the Hamiltonian should be viewed as a 
discretized model, obtained by discretizing a quantum many-body Hamiltonian in the continuous space using $N_b$ spin orbitals. In quantum chemistry, the symmetry requirement of $V$ in
Eq.~\eqref{eqn:Vsymmmetry} implies that we use the anti-symmetrized form
of the two-electron integral~\cite{SzaboOstlund1989}. For simplicity we omit the Hamiltonian in the continuous space and the detailed form of the spin orbitals, and directly focus on the discretized model. 
If an operator $\hat{O}$ commutes with the total number
operator, i.e. $[\hat{O},\hat{N}]=0$, then $\hat{O}$ is called a
particle number preserving operator. It is clear that the many-body 
Hamiltonian $\hat{H}$ is a particle number preserving operator.
In the discussion below, when the context is clear we may drop the range of the summation.


Our goal is to compute the ground-state energy of~\eqref{eqn:manybodyH}, which
can be obtained by variational principle as
\begin{equation}
  E_{0} = \inf_{\ket{\Psi} \in \mathcal{F} \,:\, \braket{\Psi\vert\Psi}  = 1}
  \bra{\Psi} \hat{H}-\mu \hat{N} \ket{\Psi}.
  \label{eqn:groundstate}
\end{equation}
The chemical potential $\mu$ is a Lagrange
multiplier chosen so that the ground state wavefunction $\vert\Psi\rangle$ has a number of electrons equal to a pre-specified integer $N\in \{0,1,\ldots, N_b\}$, such that
\begin{equation}
  \braket{\Psi\vert \hat{N} \vert \Psi}=N.
  \label{eqn:integer_condition}
\end{equation}

In the  Hartree-Fock theory, the wavefunction is not minimized with respect to the entire Fock space, but 
is assumed to be of the form
\begin{equation}\label{eqn:slater_hf}
\ket{\Psi}=\hat{c}_1^{\dagger}\cdots \hat{c}_N^{\dagger} \ket{0}.
\end{equation}
Here we defined a new set of creation operators
\begin{equation}
\hat{c}_i^{\dagger} = \sum_{p=1}^{N_b} \hat{a}_p^{\dagger} \Phi_{pi},
\quad i=1,\ldots,N.
\end{equation}
The columns of the matrix $\Phi\in\CC^{N_b\times N}$ are a set of orthonormal vectors, and the new set creation operators
satisfy the CAR
\[
\{\hat{c}_i^{\dagger},\hat{c}_j\}=\delta_{ij},\quad \{\hat{c}_i^{\dagger},\hat{c}_j^{\dagger}\}=0,\quad 
\{\hat{c}_i,\hat{c}_j\}=0.
\]
By extending $\Phi$ to be a unitary matrix $\widetilde{\Phi}$ of size $N_{b}\times N_{b}$, 
we may define $N_{b}$ creation
operators $\{\hat{c}_i^{\dagger}\}_{i=1}^{N_{b}}$, and correspondingly
the annihilation operators $\{\hat{c}_i\}_{i=1}^{N_{b}}$ that satisfy the CAR. 
Let us introduce the short hand notation
\[
  \hat{a}^{\dagger} = \left( \hat{a}_{1}^{\dagger},\ldots,
  \hat{a}_{N_{b}}^{\dagger}\right), \quad \hat{a} = \left( \hat{a}_{1},\ldots,
  \hat{a}_{N_{b}} \right)^{\top},
\]
and correspondingly
\[
  \hat{c}^{\dagger} = \left( \hat{c}_{1}^{\dagger},\ldots,
  \hat{c}_{N_{b}}^{\dagger}\right), \quad \hat{c} = \left( \hat{c}_{1},\ldots,
  \hat{c}_{N_{b}} \right)^{\top}.
\]
They satisfy
\[
\hat{c}^{\dagger} = \hat{a}^{\dagger} \widetilde{\Phi}, \quad \hat{c} =
\widetilde{\Phi}^{\dagger} \hat{a}.
\]

A wavefunction of the form \eqref{eqn:slater_hf} is called a
\textit{Slater determinant}. The term ``determinant'' comes from that the
wavefunction~\eqref{eqn:slater_hf} in the first quantized representation 
can be written as a determinant~\cite{SzaboOstlund1989}.
The set of all Slater determinants is denoted by $\mathcal{S}^{\text{HF}}$. 
It can be directly verified that any Slater determinant satisfies the constraint for the number of electrons~\eqref{eqn:integer_condition}. Hence the chemical potential only contributes a
constant term $\mu N$ to the energy and can be neglected.
Then the Hartree-Fock theory can be concisely formulated as
\begin{equation}
  E^{\text{HF}}_{0} = \inf_{\ket{\Psi} \in \mathcal{S}^{\text{HF}}} \braket{\Psi | \hat{H}  | \Psi}.
  \label{eqn:groundstate_HF}
\end{equation}
Note that each state $\ket{\Psi}\in\mathcal{S}^{\text{HF}}$ is
automatically normalized due to the orthonormal constraint of the matrix
$\Phi$.  

Using the CAR (more specifically $\hat{c}_i^{\dagger}\hat{c}_i^{\dagger}=0$), 
we have for any $\ket{\Psi}\in \mathcal{S}^{\text{HF}}$ parameterized by Eq.~\eqref{eqn:slater_hf}
\[
\hat{c}_i^{\dagger}\ket{\Psi}=0, \quad i=1,\ldots,N.
\]
In quantum physics,  the creation and
annihilation operators are often renamed as follows
\begin{equation}
  \begin{split}
    \hat{b}_i=\hat{c}_i^{\dagger}, &\quad \hat{b}^{\dagger}_i=\hat{c}_i,
  \quad  i=1,\ldots,N,\\
  \hat{b}_i=\hat{c}_i, &\quad \hat{b}_i^{\dagger}=\hat{c}_i^{\dagger}, \quad 
  i=N+1,\ldots,N_{b}.\\
  \end{split}
  \label{eqn:particle_hole_hf}
\end{equation}
We can readily verify that $\{\hat{b}_{i}^{\dagger},\hat{b}_{i}\}$ satisfy the CAR. 
Eq.~\eqref{eqn:particle_hole_hf} is often referred
to as the particle-hole transformation, where the states
$i=1,\ldots,N$ defining the Slater determinant are called ``holes'', and
the rest of the states $i=N+1,\ldots,N_{b}$ are called
``particles''~\cite{Blaizot1986}. Using the particle-hole transformation,
the Slater determinant satisfies
\begin{equation}\label{eqn:quasivacuum_hf}
  \hat{b}_i\ket{\Psi}=0, \quad i=1,\ldots,N_{b}.
\end{equation}
Recall that the property~\eqref{eqn:vacuum} for the vacuum state
$\ket{0}$, we find that a state $\ket{\Psi}\in\mathcal{S}^{\text{HF}}$
can be viewed as a vacuum operator defined by the annihilation operators
$\{\hat{b}_i\}_{i=1}^{N_{b}}$. In this sense, $\ket{\Psi}$ is called a
\textit{quasi-particle vacuum} state~\cite{Blaizot1986}.

The advantage of viewing $\ket{\Psi}$ as a quasi-particle vacuum state
is that one may concisely derive the Hartree-Fock equations using Wick's
theorem (see e.g.~\cite{Blaizot1986,NegeleOrland1988,SzaboOstlund1989}).  We omit the general formulation of
Wick's theorem here, and only give the particular instance of Wick's theorem useful for the purpose of this paper below.

\begin{prop}[Wick's theorem]\label{prop:wick}
If $\ket{\Psi}$ is a quasi-particle vacuum state with respect to a set of creation and annihilation operators satisfying the CAR, then the expectation value $\braket{ A_1 A_2 A_3 A_4} := \braket{\Psi | A_1 A_2 A_3 A_4 | \Psi}$ with respect to four operators $A_1,A_2,A_3,A_4$ can be evaluated as
\begin{equation}\label{eqn:wick}
\braket{ A_1 A_2 A_3 A_4} = \braket{ A_1 A_2 } \braket{ A_3 A_4 } - \braket{ A_1 A_3 } \braket{ A_2 A_4 } + \braket{ A_1 A_4 } \braket{ A_2 A_3 }.   
\end{equation}
\end{prop}

For any $\ket{\Psi}\in\mathcal{F}$, we may define the \textit{density matrix} as 
\begin{equation}
  \rho_{pq}:=\braket{\hat{a}_{q}^{\dagger} \hat{a}_{p}}, \quad
  p,q=1,\ldots,N_b.
  \label{eqn:DM}
\end{equation}
Then if $\ket{\Psi}\in\mathcal{S}^{\text{HF}}$, by Wick's theorem 
\begin{equation}
  \begin{split}
    \braket{\hat{a}_{p}^{\dagger} \hat{a}_{q}^{\dagger} \hat{a}_{s}
    \hat{a}_{r}} 
    =& \braket{\hat{a}_{p}^{\dagger}\hat{a}_{r}}\braket{\hat{a}_{q}^{\dagger}\hat{a}_{s}}
    -
    \braket{\hat{a}_{p}^{\dagger}\hat{a}_{s}}\braket{\hat{a}_{q}^{\dagger}\hat{a}_{r}}
    + \braket{\hat{a}_{p}^{\dagger}\hat{a}^{\dagger}_{q}}\braket{\hat{a}_{s}\hat{a}_{r}}\\
    =& \rho_{rp}\rho_{sq} - \rho_{sp}\rho_{rq}.
  \end{split}
  \label{}
\end{equation}
Note that $\ket{\Psi}$ is an eigenstate of $\hat{N}$, and hence 
the anomalous term $\braket{\hat{a}_{p}^{\dagger}\hat{a}^{\dagger}_{q}}=\braket{\hat{a}_{s}\hat{a}_{r}}=0$.

Furthermore, through direct computation we can verify
\[
  \rho_{pq} = \sum_{i=1}^{N} \Phi_{pi} \overline{\Phi}_{qi}, \quad
  p,q=1,\ldots, N_{b},
\]
or in matrix form
\[
  \rho = \Phi\Phi^{\dagger}.
\]
Therefore the density matrix for a Slater determinant is a Hermitian,
idempotent matrix and its trace is equal to $N$. In fact this defines
the set of admissible density matrices as
\begin{equation}
  \mathcal{D}^{\text{HF}} = \left\{ \rho\in \CC^{N_{b}\times N_{b}} \vert 
  \rho^{\dagger} = \rho, \quad \rho^2=\rho,\quad \Tr[\rho]=N \right\}.
  \label{}
\end{equation}


The total energy functional with respect to $\ket{\Psi}\in\mathcal{S}^{\text{HF}}$ can be written as a functional of
$\rho$ as 
\begin{equation}
  \begin{split}
    \mathcal{E}^{\text{HF}} [ \rho ] := & \braket{\hat{H}} \\ 
    = & \sum_{pq} h^{0}_{pq} \rho_{qp} +
    \frac14 \sum_{pqrs} V_{pqrs}(\rho_{rp}\rho_{sq} -
    \rho_{sp}\rho_{rq})\\
    =& \sum_{pq} h^{0}_{pq} \rho_{qp} +
    \frac12 \sum_{pqrs} V_{pqrs}\rho_{rp}\rho_{sq}.
  \end{split}
  \label{eqn:HF_Efunctional}
\end{equation}
In the last equality, we have used the anti-symmetry condition
$V_{pqrs}=-V_{pqsr}$.  

To solve the minimization problem, we define
\begin{equation}
  h_{pq}[\rho] := \frac{ \partial \mathcal{E}^{\text{HF}} }{ \partial \rho_{qp} } = 
  h^{0}_{pq} + \sum_{rs} V_{prqs} \rho_{sr}.
  \label{eqn:Hamiltonian_HF}
\end{equation}

Let $\rho$ be the minimizer of Eq.~\eqref{eqn:groundstate} with the energy functional defined in Eq. \eqref{eqn:HF_Efunctional}, and assume that 
there is a positive energy gap between the $N$-th and $(N+1)$-th eigenvalues of $h[\rho]$.  
Then the corresponding Euler-Lagrange equation  gives the following nonlinear eigenvalue problem
\begin{equation}
    h[\rho] \Phi = \Phi \Lambda,  \quad \rho=\Phi\Phi^{\dagger}.
    \label{eqn:HF_scf}
\end{equation}
Here all columns of $\Phi\in\CC^{N_b\times N}$ form an orthonormal set of $N$ vectors,
and $\Lambda=\text{diag}(\varepsilon_{1},\ldots,\varepsilon_{N})\in\RR^{N\times
N}$ is a diagonal matrix, i.e. $(\Lambda,\Phi)$ are the eigenpairs
corresponding to the lowest $N$ eigenvalues of $h[\rho]$. 
Eq.~\eqref{eqn:HF_scf} is called the Hartree-Fock equation, which needs
to be solved self-consistently until convergence. The choice of taking
the lowest $N$ eigenpairs is called the \textit{aufbau principle} in
quantum physics literature, which can be rigorously proved for certain
choice of two-particle interaction $V$~\cite{Bach1994}.

\section{Hartree-Fock-Bogoliubov theory}\label{sec:hfb}

In the Hartree-Fock theory, a state $\ket{\Psi}\in\mathcal{S}^{\text{HF}}$ has a
well defined number of electrons $N$, and the Hartree-Fock
Hamiltonian is a particle number preserving
operator. The Hartree-Fock-Bogoliubov (HFB) theory relaxes such a
constraint and allows the particle number to fluctuate. We first 
group the annihilation operators as a column vector as
\begin{equation}
  \hat{\alpha} = \left( \begin{array} { l } { \hat{a} } \\ { \hat{a}^{ \dagger } }
  \end{array} \right), 
  \label{eqn:alpha_group}
\end{equation}
i.e. $\hat{\alpha}_{p} = \hat{a}_{p} , \quad \hat{\alpha}_{N_b+p} =
\hat{a}_{p}^{\dagger} , \quad p = 1 , \ldots , N_b$, 
and the Hermitian conjugate of $\hat{\alpha}$ as $\hat{\alpha}^{ \dagger } =
\left( \hat{a}^{ \dagger }, \hat{a} \right)$. Note that $\{\hat{\alpha}^{\dagger},\hat{\alpha}\}$ are only
formal notations. In
particular, they \textit{do not} satisfy the CAR (just note that for any
$1\le p\le N_{b}$, $\{\hat{\alpha}_{p},\hat{\alpha}_{p+N_{b}}\} =
\{\hat{a}_{p},\hat{a}_{p}^{\dagger}\} = 1 \ne 0$).

In parallel to the Hartree-Fock theory, the ground state
wavefunction $\ket{\Psi}$ in the HFB theory is assumed to be a \textit{generalized}
Slater determinant
\begin{equation}\label{eqn:slater_hfb}
  \ket{\Psi}=\hat{c}_1^{\dagger}\cdots \hat{c}_{N_{b}}^{\dagger} \ket{0}.
\end{equation}
Note that there are $N_{b}$ instead of $N$ creation operators acting on
the vacuum state. In fact, we
shall define $N_{b}$ creation operators
\begin{equation}\label{eqn:creation_hfb}
\hat{c}_i^{\dagger} = \sum_{p=1}^{2N_b} \hat{\alpha}_p^{\dagger} \Phi_{pi},
\quad i=1,\ldots,N_{b}.
\end{equation}
We shall demonstrate below that by a proper choice of $\Phi\in
\CC^{2N_{b}\times N_{b}}$, the new
creation and annihilation operators
$\{\hat{c}^{\dagger}_{i},\hat{c}_{i}\}_{i=1}^{N_{b}}$ indeed 
satisfy the CAR.  If so, we may perform a new (and trivial)
particle-hole transformation as
\begin{equation}
    \hat{b}_i=\hat{c}_i^{\dagger}, \quad \hat{b}^{\dagger}_i=\hat{c}_i
    \quad  i=1,\ldots,N_{b}.
  \label{eqn:particle_hole_hfb}
\end{equation}
Then $\{\hat{b}^{\dagger}_{i},\hat{b}_{i}\}_{i=1}^{N_{b}}$ satisfy the
CAR, and the generalized Slater determinant~\eqref{eqn:slater_hfb} satisfy
\begin{equation}
  \hat{b}_{i}\ket{\Psi}=0, \quad i=1,\ldots, N_{b},
  \label{}
\end{equation}
i.e. it is a quasi-particle vacuum state. The set of all such generalized Slater determinants is denoted by
$\mathcal{S}^{\text{HFB}}$.

We define the density matrix $\rho$, and the pair matrix $\kappa$ 
associated with a state $\ket{\Psi}$ as
\begin{equation}
  \rho_{pq}:=\braket{\hat{a}_{q}^{\dagger} \hat{a}_{p}}, \quad
  \kappa_{pq}:=\braket{\hat{a}_{q} \hat{a}_{p}}, \quad 
  p,q=1,\ldots,N_b.
  \label{eqn:dm_pair}
\end{equation}
Note that $\ket{\Psi}\in \mathcal{S}^{\text{HFB}}$ does not necessarily
have a well defined particle number, i.e. $\ket{\Psi}$ is not an
eigenstate of $\hat{N}$, and the pair matrix $\kappa$ may not vanish.
From the CAR we find that the density matrix and the pair matrix satisfy
the symmetry properties
\begin{equation}
  \rho^{\dagger} = \rho, \quad  \kappa^{\top} = - \kappa.
  \label{}
\end{equation}
We also define generalized density matrix
\begin{equation}
  R_{pq} := \braket{\hat{\alpha}_{q} ^ {\dagger} \hat{\alpha}_{p}},
  \label{}
\end{equation}
and its block structure is recorded in Proposition \ref{prop:dm_hfb}.

\begin{prop}\label{prop:dm_hfb}
  In a matrix block form, the generalized density matrix $R$ can be written  as
  \begin{equation}\label{eqn:GDM_form}
    R = \left( \begin{array} { c c } { \rho } & { \kappa } \\ { - \overline
      \kappa} & {I-  \overline \rho} \end{array} \right).
  \end{equation}
  where $I$ is the identity matrix of size $N_{b}$.  Furthermore, $R$ is
Hermitian.
\end{prop}
\begin{proof}
  According to the block partition of $\hat{\alpha}$, let us first
  partition $R$ in a $2\times 2$ matrix block form as 
  \begin{equation}\label{eqn:GDM_form_tmp}
    R = \left( \begin{array} { c c } { A } & { B } \\ {
    C} & {D} \end{array} \right).
  \end{equation}
  Then clearly the $(1,1)$ matrix block $A=\rho$, and the $(1,2)$ matrix
  block $B=\kappa$. The $(2,1)$ block is
  \[
  C_{pq} = \braket{\hat{a}^{\dagger}_{q} \hat{a}^{\dagger}_{p}} =
  \overline{\braket{\hat{a}_{p} \hat{a}_{q}}} =
  \overline{\kappa}_{qp}.
  \]
  Using the condition $\kappa=-\kappa^{\top}$, we have
  $\overline{\kappa}_{qp} = -\overline{\kappa}_{pq}$, and hence
  $C=-\overline{\kappa}$.  
  
  Similarly, the $(2,2)$-block can be computed as
  \[
  D_{pq} = \braket{\hat{a}_{q} \hat{a}^{\dagger}_{p}} =
  \delta_{pq} - \braket{\hat{a}^{\dagger}_{p} \hat{a}_{q}} = \delta_{pq}
  - \rho_{qp}.
  \]
  Since $\rho=\rho^{\dagger}$, we have $\rho_{qp}=\overline{\rho}_{pq}$, and
  hence $D=I-\overline{\rho}$. This proves the
  form~\eqref{eqn:GDM_form}.

  Again using symmetry properties of $\rho,\kappa$, we have
  \[
    R^{\dagger} = \left( \begin{array} { c c } { \rho^{\dagger} } & {
      -\kappa^{\top} } \\ { \kappa^{\dagger} } & {I-  \overline
      \rho^{\dagger}} \end{array} \right) = \left( \begin{array} { c c } { \rho } & { \kappa } \\ { - \overline
      \kappa} & {I-  \overline \rho} \end{array} \right) = R. 
    \label{}
  \]
  Hence $R$ is Hermitian.
\end{proof}

For $\ket{\Psi}\in\mathcal{S}^{\text{HFB}}$,   we can apply Wick's theorem
in Proposition \ref{prop:wick} again to evaluate expectation values
\begin{equation}
  \begin{split}
    \braket{\hat{a}_{p}^{\dagger} \hat{a}_{q}^{\dagger} \hat{a}_{s}
    \hat{a}_{r}} 
    =& \braket{\hat{a}_{p}^{\dagger}\hat{a}_{r}}\braket{\hat{a}_{q}^{\dagger}\hat{a}_{s}}
    - \braket{\hat{a}_{p}^{\dagger}\hat{a}_{s}}\braket{\hat{a}_{q}^{\dagger}\hat{a}_{r}}
    + \braket{\hat{a}_{p}^{\dagger}\hat{a}^{\dagger}_{q}}\braket{\hat{a}_{s}\hat{a}_{r}} 
    \\
    =& \rho_{rp}\rho_{sq} - \rho_{sp}\rho_{rq} + \overline{\kappa}_{pq}\kappa_{rs}.
  \end{split}
  \label{}
\end{equation}
Due to the breaking of the particle number symmetry, there is an anomalous term in the energy defined in terms of the pair matrix
$\kappa$ and its conjugate. 
Furthermore, $R$ can be written as
\[
  R_{pq} = \sum_{i=1}^{N_{b}} \Phi_{pi} \overline{\Phi}_{qi}, \quad
  p,q=1,\ldots, 2N_{b},
\]
or in matrix form
\[
  R = \Phi\Phi^{\dagger}.
\]
Hence $R$ is an idempotent matrix, and we always have $\Tr[R]=N_{b}$. In
the HFB theory, the number of electrons cannot be set by the
rank condition of the generalized density matrix $R$. Instead  it is givne by the trace of the density matrix $\rho$, i.e.
$\Tr[\rho]=N$.  The set of the generalized density matrix is defined as 
\begin{equation}
  \mathcal{D}^{\text{HFB}} = \left\{ R\in \CC^{2N_{b}\times 2N_{b}} \vert 
  R^{\dagger} = R, \quad R^2=R,\quad \Tr[\rho]=N \right\}.
  \label{}
\end{equation}

The total energy functional of the HFB theory is
\begin{equation}
  \begin{split}
    \mathcal{E}^{\text{HFB}} [R] := & 
    \braket{\hat{H}-\mu \hat{N}} \\ 
    = & \sum_{pq} (h^{0}_{pq}-\mu \delta_{pq}) \rho_{qp} +
    \frac14 \sum_{pqrs} V_{pqrs}(\rho_{rp}\rho_{sq} -
    \rho_{sp}\rho_{rq}+ \overline{\kappa}_{pq}\kappa_{rs})\\
    =& \sum_{pq} (h^{0}_{pq}-\mu \delta_{pq}) \rho_{qp} +
    \frac12 \sum_{pqrs} V_{pqrs}\rho_{rp}\rho_{sq} +
    \frac14 \sum_{pqrs} V_{pqrs}\overline{\kappa}_{pq}\kappa_{rs}.
  \end{split}
  \label{eqn:HFB_Efunctional}
\end{equation}

Let $R$ be the minimizer of the following problem
\begin{equation}
  E^{\text{HFB}}_{0} = \inf_{R\in \mathcal{D}^\text{HFB}}
  \mathcal{E}^{\text{HFB}} [ R ].
  \label{eqn:HFB_DM_minimize}
\end{equation}
with its energy functional defined in Eq.~\eqref{eqn:HFB_Efunctional}. Similar to Eq.~\eqref{eqn:Hamiltonian_HF}, the linearized Hamiltonian
due to variation with respect to $\rho$ is (the only addition is a
diagonal term
due to the chemical potential)
\begin{equation}
  h_{pq}[\rho] := \frac{\partial \mathcal{E}^{\text{HFB}}[R]}{ \partial \rho_{qp} } = 
  h^{0}_{pq}-\mu \delta_{pq} + \sum_{rs} V_{prqs} \rho_{sr}.
  \label{eqn:Hamiltonian_HFB}
\end{equation}
Due to the pair matrix, we also define the pairing field (also called
the pairing potential, or gap function in the setting of translation-invariant systems) as
\begin{equation}
  \Delta_{pq}[\kappa] := \frac { \partial \mathcal{E}^{\text{HFB}}[R] } { \partial \overline
\kappa_{pq} } = \frac14 \sum_{rs} V_{pqrs} \kappa_{rs} - \frac14
\sum_{rs} V_{qprs} \kappa_{rs} = \frac12 \sum_{rs} V_{pqrs} \kappa_{rs}.
  \label{eqn:HFB_delta}
\end{equation}
Note the difference between the order of
the $p,q$ indices in $h$ and $\Delta$. Here we have used the anti-symmetry
property of $\kappa$ (and hence $\overline{\kappa}$), as well as the
anti-symmetry property of $V$. 
Similar to $\rho,\kappa$, the matrices
$h[\rho],\Delta[\kappa]$ satisfy the symmetry properties
\begin{equation}
  h=h^{\dagger}, \quad \Delta=-\Delta^{\top}.
  \label{eq:delta_property}
\end{equation}

Then the Euler-Lagrange equation corresponding to the minimization problem \eqref{eqn:HFB_DM_minimize} gives the following nonlinear eigenvalue problem
\begin{equation}
  \mathscr{H}[R] \Phi = \Phi \Lambda,  \quad R=\Phi\Phi^{\dagger}.
  \label{eqn:HFB_scf}
\end{equation}
Here the quasi-particle Hamiltonian $\mathscr{H}[R]$ is defined as
\begin{equation}
 \mathscr{H}[R] = \left( \begin{array} { c c } { h[\rho] } & {
   \Delta[\kappa] } \\ { -
   \overline{\Delta}[\kappa]} & { - \overline{h}[\rho]} \end{array}
   \right).
  \label{eqn:HFB_H}
\end{equation}
$\mathscr{H}[R]$ is a Hermitian matrix due to the symmetry properties of $h,\Delta$. The eigenvalues are ordered non-decreasingly as 
$\varepsilon_1\le \varepsilon_2 \ldots\le \varepsilon_{N_b}$, and 
$\Lambda=\text{diag}(\varepsilon_1,\ldots,\varepsilon_{N_b})$. The eigenvectors $\Phi\in\CC^{2 N_b\times N_{b}}$ are an orthonormal set of
$N_{b}$ vectors, and $(\Lambda,\Phi)$ are the eigenpairs
corresponding to the lowest $N_{b}$ eigenvalues of
$\mathscr{H}[R]$. The chemical potential $\mu$ should be adjusted
so that 
\begin{equation}
  N = \langle \Psi  \vert  \hat{N}  \vert  \Psi \rangle = \Tr[\rho].
  \label{eqn:Ne_condition}
\end{equation}
Again, the \textit{aufbau principle}, or the
choice of taking the algebraically lowest $N_{b}$ eigenpairs, does not
always hold, but the principle can be rigorously proved for certain
choice of two-particle interaction $V$~\cite{Bach1994,LewinPaul}. 
The eigenvalue problem~\eqref{eqn:HFB_scf} is called the Hartree-Fock-Bogoliubov (HFB)
equation, which is also called the Bogoliubov-de Genns (BdG) equation.

\section{Solving Hartree-Fock-Bogoliubov equations}\label{sec:solvehfb}

Similar to Hartree-Fock equations, the numerical solution of HFB equations is often obtained using self-consistent field iterations. At the $\ell$-th iteration, the density matrix, pair matrix are denoted by $\rho^{(\ell)},\kappa^{(\ell)}$, respectively. We may then construct $R^{(\ell)}$ and the quasi-particle Hamiltonian $\mathscr{H}[R^{(\ell)}]$. The lowest $N_b$ eigenpairs of $\mathscr{H}[R^{(\ell)}]$ further can be used to define $\rho^{(\ell+1)},\kappa^{(\ell+1)}$, and may continue the iteration until reaching self-consistency.  Then we may compute the total energy using Eq.~\eqref{eqn:HFB_Efunctional}. 
We discuss the diagonalization method for solving HFB equations for general quasi-particle Hamiltonians in Section \ref{sec:generalhfb}.  In practice, the quasi-particle Hamiltonian often exhibits additional sparsity pattern due to the separation between spatial and spin degrees of freedom. In such a case, the Hamiltonian can be block-diagonalized. This will be discussed in Section \ref{sec:symmetry}.

\subsection{General formulation}\label{sec:generalhfb}
The quasi-particle Hamiltonian $\mathscr{H}$ is a highly structured matrix. The corresponding structure of its eigenpairs is recorded in Proposition \ref{prop:hfb_eigen}.
\begin{prop}\label{prop:hfb_eigen}
The eigenvalues of $\mathscr{H}$ are real and symmetric with respect to
$0$. Assuming there is no zero eigenvalue, $\mathscr{H}$
can be diagonalized in the following form
\begin{equation}
  \mathscr{H} \left( \begin{array} { c c } { U } & { \overline { V } }
    \\ { V } & { \overline { U } } \end{array} \right) = \left(
    \begin{array} { c c } { U } & { \overline { V } } \\ { V } & {
      \overline { U } } \end{array} \right) \left( \begin{array} { c c }
        { \Lambda } & { 0 } \\ { 0 } & { - \Lambda } \end{array}
        \right). 
  \label{eqn:HFB_diag}
\end{equation}
Here $\Lambda\in\RR^{N_b\times N_b}$ is a diagonal matrix with positive entries, and the
matrix 
\[
\left( \begin{array} { c c } { U } & { \overline { V } }
  \\ { V } & { \overline { U } } \end{array} \right)\in U(2N_b).
\]
\end{prop}
\begin{proof}
  Without loss of generality let $(u^{\top},v^{\top})^{\top}$ be an eigenvector with
  eigenvalue $\lambda>0$, i.e.
  \[
  \mathscr{H} \begin{pmatrix}u\\ v\end{pmatrix} =
    \lambda\begin{pmatrix}u\\ v\end{pmatrix},
  \]
  or 
  \[
  hu + \Delta v = \lambda u, \quad -\overline{h} v -\overline{\Delta} u
  = \lambda v.
  \]
  Taking the complex conjugate and negating both equations, we have
  \[
  -\overline{h} \overline{u} - \overline{\Delta} \overline{v} = -\lambda
  \overline{u}, \quad h \overline{v} +\Delta \overline{u} = \lambda \overline{v}.
  \]
  In the matrix form, this becomes
  \[
  \mathscr{H} \begin{pmatrix}\overline{v}\\ \overline{v}\end{pmatrix} =
    -\lambda\begin{pmatrix}\overline{v}\\ \overline{v}\end{pmatrix}.
  \]
  This holds for every eigenpair, and we prove the form of the
  decomposition~\eqref{eqn:HFB_diag}. The eigenvectors form a unitary matrix
  directly follow from that $\mathscr{H}$ is Hermitian.
\end{proof}

Due to Proposition~\ref{prop:hfb_eigen}, the lowest $N_b$ eigenvalues are always non-positive and are denoted by 
$$
\Lambda=\text{diag}(\varepsilon_{1},\ldots,\varepsilon_{N_{b}}), \quad \varepsilon_i\le 0.
$$ 
Then according to Eq.~\eqref{eqn:HFB_scf}, we find that
\[
\Phi = \left( \begin{array} { l } { \overline{V} } \\ { \overline{U}}  \end{array}\right).
\]
The generalized density matrix can be constructed from the eigenvectors associated with negative eigenvalues as
\begin{equation}
  R = \begin{pmatrix} 
    \overline { V } \\ \overline { U }   \end{pmatrix}
  \begin{pmatrix} V^{\top} & U^{\top}  \end{pmatrix}.
  \label{eqn:gdm}
\end{equation}
Compared to Eq.~\eqref{eqn:GDM_form}, we find that 
\begin{equation}
  \rho = \overline{V} V^{\top}, \quad \kappa = \overline{V} U^{\top}.
  \label{}
\end{equation}

It remains to show that the creation and annihilation operators defined in Eq.~\eqref{eqn:creation_hfb} satisfy the CAR. Rewrite
Eq.~\eqref{eqn:creation_hfb} as
\[
\hat{c}_i^{\dagger} = \sum_{p=1}^{N_b} \hat{\alpha}_p^{\dagger}
\overline{V}_{pi} + \sum_{q=1}^{N_{b}} \hat{\alpha}_q \overline{U}_{qi},
\quad i=1,\ldots,N_{b}.
\]
Correspondingly the annihilation operators are defined as
\[
\hat{c}_i = \sum_{p=1}^{N_b} \hat{\alpha}_p
V_{pi} + \sum_{q=1}^{N_{b}} \hat{\alpha}_q^{\dagger} U_{qi},
\quad i=1,\ldots,N_{b}.
\]
If we group $\begin{pmatrix} \hat{c}^{\dagger} & \hat{c} \end{pmatrix}$
as a row vector, just as $\hat{\alpha}^{\dagger}$, then we have the
matrix form
\begin{equation}
  \begin{pmatrix} \hat{c}^{\dagger} & \hat{c} \end{pmatrix} =
  \begin{pmatrix} \hat{a}^{\dagger} & \hat{a} \end{pmatrix} 
  \begin{pmatrix} 
    \overline { V } & U  \\ \overline{U} & V  
  \end{pmatrix},
  \label{eqn:hfb_op_matrix}
\end{equation}
which is unitary transformation of the fermionic creation and
annihilation operators from Proposition \ref{prop:hfb_eigen}. Therefore  $\{\hat{c}^{\dagger}_{i},\hat{c}_{i}\}_{i=1}^{N_{b}}$ indeed satisfy the
CAR, 
$\ket{\Psi}\in\mathcal{S}^{\text{HFB}}$ is  a quasi-particle
vacuum, and Wick's theorem is applicable.


Due to the symmetry of eigenvalues according to Proposition \ref{prop:hfb_eigen}, the existence of a positive gap $\Delta_{g}:=\varepsilon_{N_{b}+1}-\varepsilon_{N_{b}}$ is equivalent to the statement that $\mathscr{H}$ does not have a zero eigenvalue. Then the generalized density matrix can be compactly written using a matrix function as
\[
R = \mathbbm{1}_{(-\infty,0)}(\mathscr{H}),
\]
where $\mathbbm{1}_{(-\infty,0)}(\cdot)$ is the indicator function on $(-\infty,0)$. 

In the case when $\mathscr{H}$ has zero eigenvalues or if the energy gap $\Delta_{g}$ is small, we need to employ the finite temperature formulation.  Let $\beta=T^{-1}$ be the inverse temperature (the Boltzmann constant is taken to be $1$), then the generalized density matrix should be weighted by the Fermi-Dirac distribution as
\[
R=f_{\beta}(\mathscr{H}):=\frac{1}{e^{\beta\mathscr{H}}+1}.
\]
In the finite temperature formulation, $R$ is the minimizer of the Helmholtz free energy, defined as
\begin{equation}
  \mathcal{F} = E - \frac{1}{\beta} S.
  \label{eqn:free_energy}
\end{equation}
Here 
\begin{equation}
  S = -\sum_{i=1}^{2N_{b}} [f_{i} \log f_{i} + (1-f_{i})\log (1-f_{i})]
  \label{eqn:entropy}
\end{equation}
is the entropy, and $f_{i} = (1+e^{\beta \varepsilon_{i}})^{-1}$ is the Fermi-Dirac distribution.

When solving HFB equations self-consistently, we should note that there are in fact ``two chemical potentials''. One chemical potential is used to separate the wanted and unwanted eigenpairs. This chemical potential is always set to $0$ due to Proposition \ref{prop:hfb_eigen}. The other chemical potential, or the ``true'' chemical potential denoted by $\mu$, is used to control the number of electrons so that Eq.~\eqref{eqn:Ne_condition} is satisfied. These two chemical potentials coincide in standard Hartree-Fock type of calculations, and can be identified directly after a single step of diagonalization.  In HFB, the number of electrons cannot be exactly controlled even if we fully diagonalize the system as in Eq.~\eqref{eqn:HFB_diag}.  This is because the total number operator $\hat{N}$ does not commute with $\hat{H}$. 
Hence the constraint on the number of electrons can only be satisfied by
dynamic adjustment of the chemical potential $\mu$.  For clarity, we summarize the differences between HF and HFB in Table \ref{tab:compare_hf_hfb}, when diagonalization type methods are used. The reason why HFB is not suitable for iterative solver is that the number of eigenpairs to compute is always $N_b$, i.e. half of the eigenpairs. In this regime, Krylov type iterative solvers are \textit{not} efficient.
\begin{table}[!htp]
  \centering
  \begin{tabular}{c|c|c}
    \toprule
    & HF & HFB\\
    \hline
    Matrix size & $N_{b}$ & $2N_{b}$\\
    \hline
    Number of eigenpairs to compute & $N$ & $N_{b}$\\
    \hline
    Number of ``holes'' & $N$ & $N_{b}$\\
    \hline
    Number of ``particles'' & $N_{b}-N$ & $N_{b}$\\
    \hline
    Fixed number of particles & Yes & No \\
    \hline
    Directly obtain $\mu$ after diagonalization& Yes & No \\
    \hline
    Good for iterative methods &  Yes & No\\
    \bottomrule
  \end{tabular}
  \caption{Comparison between parameters for solving Hartree-Fock theory and
  Hartree-Fock-Bogoliubov theory based on a diagonalization procedure.}
  \label{tab:compare_hf_hfb}
\end{table}

%

\subsection{Spin symmetry}\label{sec:symmetry}
In the previous discussion, each site $p$ refers to a general spin orbital. In practice $p$ is often resolved into a spatial index $i$, and a spin index $\sigma\in\{\uparrow,\downarrow\}$. We denote by $\widetilde{N}_{b}=N_{b}/2$ the number of spatial orbitals. We further assume that 
partitioned according to the spin index, $h$ is a $2\times 2$ block diagonal matrix as
\begin{equation}
  h =
  \begin{pmatrix}
    h_{\uparrow} & 0\\
    0 & h_{\downarrow}
  \end{pmatrix}.
  \label{eqn:h_block}
\end{equation}
In other words, we assume that there is no spin-orbit coupling effect in $h$.  The pairing field can be expressed generally in the matrix block form as
\begin{equation}
  \label{eqn:Delta_block}
  \Delta =
  \begin{pmatrix}
    \Delta_{\uparrow\uparrow} & \Delta_{\uparrow\downarrow}\\
    \Delta_{\downarrow\uparrow} & \Delta_{\downarrow\downarrow}
  \end{pmatrix}.
\end{equation}
Define the grouped creation operator
$\hat{\alpha}^{\dagger}=(\hat{a}^{\dagger}_{\uparrow},\hat{a}^{\dagger}_{\downarrow},\hat{a}_{\uparrow},\hat{a}_{\downarrow})$,
and let $\hat{\alpha}$ be the corresponding annihilation operator, then
the corresponding quasi-particle Hamiltonian takes the form
\begin{equation}
  \mathscr{H} =
  \left(
    \begin{array} {cccc}
      { h_{\uparrow} } & { 0 } & { \Delta_{\uparrow\uparrow} } & {\Delta_{\uparrow\downarrow }} \\
      { 0 } & { h_{\downarrow} } & {\Delta_{\downarrow\uparrow}} & {\Delta_{\downarrow\downarrow} } \\
      { - \overline{\Delta}_{\uparrow\uparrow}  } & { - \overline{\Delta}_{\uparrow\downarrow} } & {-\overline{h}_{\uparrow}  } & { 0 } \\
      { - \overline{\Delta}_{\downarrow\uparrow}  } & { - \overline{\Delta}_{\downarrow\downarrow}  } & { 0 } & { - \overline{h}_{\downarrow} }
    \end{array}
  \right).
  \label{eqn:HFB_special}
\end{equation}

The pairing field $\Delta$ can be generally decomposed as
\begin{equation}
  \Delta = \sum_{\alpha=0}^{3}\chi_\alpha \otimes \Delta_\alpha = 
  \begin{pmatrix}
    -\Delta_1+\I\Delta_2 & \Delta_{0}+\Delta_3\\
    -\Delta_{0}+\Delta_3 & \Delta_1+\I\Delta_2
  \end{pmatrix},
  \label{eqn:Delta_decompose}
\end{equation}
where
\begin{equation}
\chi_0=\begin{pmatrix}
  0 & 1\\
  -1 & 0
\end{pmatrix}, \quad 
  \chi_1=
  \begin{pmatrix}
    -1 & 0\\0 & 1
  \end{pmatrix}\text{, } 
  \chi_2=
  \begin{pmatrix}
    \I & 0\\0 & \I
  \end{pmatrix},
  \quad 
  \chi_3 =
  \begin{pmatrix}
    0 & 1\\1 & 0
  \end{pmatrix}.
  \label{eqn:chi_matrix}
\end{equation}
The matrices $\{\chi_\alpha\}$ are related to the Pauli matrices, and its convention follows that in~\cite{zhu2016bogoliubov}.  Since $\Delta$ should be an anti-symmetric matrix, we have
\begin{equation}
  \Delta_0=\Delta_{0}^{\top}, \quad \Delta_1 = -\Delta_1^\top, \quad \Delta_2= -\Delta_2^\top, \quad \Delta_3 = -\Delta_3^\top.
  \label{}
\end{equation}

When only $\Delta_{0}$ is present, it is called the \textit{spin-singlet coupling} regime,  and the total spin satisfies $S_{z}=0$. This is the most common regime, as is found in conventional superconductors, iron-based superconductors etc. Correspondingly, when $\Delta_{0}$ vanishes and only $\Delta_{1},\Delta_{2},\Delta_{3}$ are present, it is called the \textit{spin-triplet coupling} regime, and the total spin satisfies $S_{z}=1$. The spin-triplet coupling is less common, and has been found in $^3$He superfluid as well as certain heavy-fermion superconductors~\cite{zhu2016bogoliubov}.

In order to reduce the dimension of $\mathscr{H}$, following the structure in~\eqref{eqn:HFB_special}, $\Delta$ needs to be either a block-diagonal matrix, or a block-off-diagonal matrix. 

{\bf $\Delta$ is a block-off-diagonal matrix}. In this case, only $\Delta_0,\Delta_3$ can be present. 
The quasi-particle Hamiltonian~\eqref{eqn:HFB_special} is then block diagonalized in terms of the following two Hermitian matrix blocks.
\begin{equation}
  \mathscr{H}_{\uparrow\downarrow}:=
  \begin{pmatrix}
    h_{\uparrow}  &  \Delta_0+\Delta_3\\
    \overline{\Delta}_0-\overline{\Delta}_3 & - \overline{h}_{\downarrow}
  \end{pmatrix}, \quad  \mathscr{H}_{\downarrow\uparrow}:= 
  \begin{pmatrix} h_{\downarrow}  & -\Delta_0+\Delta_3\\
    -\overline{\Delta}_0-\overline{\Delta}_3       & - \overline{h}_{\uparrow}
  \end{pmatrix}.
  \label{eq:hfb_singlet}
\end{equation}
Define $J=\begin{pmatrix} 0 & I_{\widetilde N_{b}} \\ I_{\widetilde N_{b}} & 0\end{pmatrix}$, then we have the relation
\[
  \mathscr{H}_{\downarrow\uparrow} = -J
  \overline{\mathscr{H}}_{\uparrow\downarrow} J.
\]
This suggests that for an eigenpair of $\mathscr{H}_{\uparrow\downarrow}$, 
\begin{equation}
  \label{eq:3}
  \mathscr{H}_{\uparrow\downarrow} \begin{pmatrix}u\\v\end{pmatrix}=\lambda
  \begin{pmatrix}u\\v\end{pmatrix}, \quad \lambda\in \RR,
\end{equation}
we have
\begin{equation}
  \label{eq:4}
  \mathscr{H}_{\downarrow\uparrow} \begin{pmatrix}\overline v\\\overline u    
  \end{pmatrix}=-\lambda
  \begin{pmatrix}\overline v\\\overline u \end{pmatrix}.
\end{equation}
Therefore the negative eigenvalues of $\mathscr{H}_{\uparrow\downarrow}$ can be mapped to the positive eigenvalues of $\mathscr{H}_{\downarrow\uparrow}$, and vice versa. This is in fact a direct corollary of Proposition \ref{prop:hfb_eigen}. 
In order to solve HFB, we only need to diagonalize
$\mathscr{H}_{\uparrow\downarrow}$ as 
\begin{equation}
  \label{eq:14}
  \mathscr{H}_{\uparrow\downarrow}
  \begin{pmatrix}
    U_{\uparrow\downarrow} & \overline{V}_{\downarrow\uparrow}\\
    V_{\uparrow\downarrow} & \overline{U}_{\downarrow\uparrow}\\
  \end{pmatrix}
  =
  \begin{pmatrix}
    U_{\uparrow\downarrow} & \overline{V}_{\downarrow\uparrow}\\
    V_{\uparrow\downarrow} & \overline{U}_{\downarrow\uparrow}\\
  \end{pmatrix}
  \begin{pmatrix}
    \Lambda_{\uparrow\downarrow}^+ & 0\\ 0 & -\Lambda_{\downarrow\uparrow}^+
  \end{pmatrix}
\end{equation}
where $\Lambda_{\uparrow\downarrow}^+$ and $\Lambda_{\downarrow\uparrow}^+$ are the diagonal matrices formed by the positive eigenvalues of $\mathscr{H}_{\uparrow\downarrow}$ and $\mathscr{H}_{\downarrow\uparrow}$.  Note that in general, the number of negative
eigenpairs contained in $\Lambda_{\downarrow\uparrow}^+$ may not be equal to $\widetilde{N}_{b}$. Accordingly, $U_{\uparrow\downarrow}, V_{\uparrow\downarrow}, U_{\downarrow\uparrow}, V_{\downarrow\uparrow}$ are not necessarily square matrices.  
The positive eigenpairs of $\mathscr{H}_{\uparrow\downarrow}$ can be mapped to the negative eigenpairs of
$\mathscr{H}_{\downarrow\uparrow}$. Hence, the total number of negative
eigenpairs of $\mathscr{H}$ remains $N_{b}$.
Each block has the same dimension. The density matrix associated with $\mathscr{H}$ is constructed as
\begin{equation}
  \label{eq:15}
    R = 
    \begin{pmatrix}
      \overline{V}_{\downarrow\uparrow}V_{\downarrow\uparrow}^\top & 0 & 0 & \overline{V}_{\downarrow\uparrow}U_{\downarrow\uparrow}^\top\\
      0 & \overline{V}_{\uparrow\downarrow}V_{\uparrow\downarrow}^\top &  \overline{V}_{\uparrow\downarrow}U_{\uparrow\downarrow}^\top & 0\\
      0 & \overline{U}_{\uparrow\downarrow}V_{\uparrow\downarrow}^\top & \overline{U}_{\uparrow\downarrow}U_{\uparrow\downarrow}^\top & 0\\
      \overline{U}_{\downarrow\uparrow}V_{\downarrow\uparrow}^\top & 0 & 0 & \overline{U}_{\downarrow\uparrow}U_{\downarrow\uparrow}^\top
    \end{pmatrix}
\end{equation}

Then the density matrix is  $$\rho = \begin{pmatrix}
  \overline{V}_{\downarrow\uparrow}V_{\downarrow\uparrow}^\top & 0\\
  0 & \overline{V}_{\uparrow\downarrow}V_{\uparrow\downarrow}^\top
\end{pmatrix}.
$$ 
Due to the orthogonality of the eigenvectors of $\mathscr{H}_{\uparrow\downarrow}$, the relation 
$$I-\overline{\rho}=
\begin{pmatrix}
  \overline{U}_{\uparrow\downarrow}U_{\uparrow\downarrow}^\top & 0\\
  0 & \overline{U}_{\downarrow\uparrow}U_{\downarrow\uparrow}^\top
\end{pmatrix}$$
is indeed satisfied. We may define the spin-up matrix $\rho_\uparrow=\overline{V}_{\downarrow\uparrow}V_{\downarrow\uparrow}^\top$, the spin-down density matrix  $\rho_\downarrow=\overline{V}_{\uparrow\downarrow}V_{\uparrow\downarrow}^\top$, and the pair matrix $\kappa_{\uparrow\downarrow}=\overline{V}_{\downarrow\uparrow}U_{\downarrow\uparrow}^{\top}$. 
The (spin) reduced density matrix is
\begin{equation}
  P = \begin{pmatrix} \rho_{\uparrow}  & \kappa_{\uparrow\downarrow}\\
    \kappa_{\uparrow\downarrow}^{\dagger} & I -
    \overline{\rho}_{\downarrow},
  \end{pmatrix}.
  \label{eq:gdm_singlet}
\end{equation}
The total number of electrons should satisfy
\begin{equation}
  N = \Tr[\rho_{\uparrow}] + \Tr[\rho_{\downarrow}].
  \label{}
\end{equation}
In particular, the number of electrons \textit{cannot} be evaluated as
$\Tr[P]$ as in standard Hartree-Fock calculations.

{\bf $\Delta$ is a block-diagonal matrix.}
In this case, the quasi-particle Hamiltonian \eqref{eqn:HFB_special} can be decomposed into the two following Hermitian matrices:
\begin{equation}
  \label{eq:10}
    \mathscr{H}_{\uparrow\uparrow}:=
  \begin{pmatrix}
    h_{\uparrow}  &  -\Delta_1+\I \Delta_2\\
    -\Delta_1^{\dagger}-\I\Delta_2^\dagger & -\overline{h}_{\uparrow}
  \end{pmatrix}
  \text{ and }
  \mathscr{H}_{\downarrow\downarrow}:=
  \begin{pmatrix}
    h_{\downarrow}  &  \Delta_1+\I\Delta_2\\
    \Delta_1^{\dagger}-\I\Delta_2^\dagger & -\overline{h}_{\downarrow}
  \end{pmatrix}
\end{equation}
Unlike the block-diagonal case, $\mathscr{H}_{\uparrow\uparrow}$ and $    \mathscr{H}_{\downarrow\downarrow}$ may not be directly related. However, these two matrices themselves can be diagonalized into the following two forms
\begin{equation}
  \label{eq:11}
  \mathscr{H}_{\uparrow\uparrow}
  \begin{pmatrix}
    U_{\uparrow\uparrow} & \overline V_{\uparrow\uparrow}\\
    V_{\uparrow\uparrow} & \overline U_{\uparrow\uparrow}
  \end{pmatrix}
  = 
  \begin{pmatrix}
    U_{\uparrow\uparrow} & \overline V_{\uparrow\uparrow}\\
    V_{\uparrow\uparrow} & \overline U_{\uparrow\uparrow}
  \end{pmatrix}
  \begin{pmatrix}
    \Lambda_\uparrow & 0\\
    0 & -\Lambda_\uparrow
  \end{pmatrix}
\end{equation}
and
\begin{equation}
  \mathscr{H}_{\downarrow\downarrow}
  \begin{pmatrix}
    U_{\downarrow\downarrow} & \overline V_{\downarrow\downarrow}\\
    V_{\downarrow\downarrow} & \overline U_{\downarrow\downarrow}
  \end{pmatrix}
  = 
  \begin{pmatrix}
    U_{\downarrow\downarrow} & \overline V_{\downarrow\downarrow}\\
    V_{\downarrow\downarrow} & \overline U_{\downarrow\downarrow}
  \end{pmatrix}
  \begin{pmatrix}
    \Lambda_\downarrow & 0\\
    0 & -\Lambda_\downarrow
  \end{pmatrix}
\end{equation}
By Proposition \ref{prop:hfb_eigen}, each eigenvalue problem has $\widetilde N_b$ positive eigenvalues and $\widetilde N_b$ negative eigenvalues, respectively, and all the blocks of eigenvector matrices have the same size $\widetilde N_b$. The reduced density matrices $P_{\uparrow}$ and $P_{\downarrow}$ can also be constructed separately,
\begin{equation}
  \label{eq:12}
  P_{\uparrow} =
  \begin{pmatrix}
    \rho_\uparrow & \kappa_{\uparrow\uparrow}\\
    \kappa_{\uparrow\uparrow}^\dagger & I-\overline{\rho}_\uparrow
  \end{pmatrix}
  \text{ and }
  P_{\downarrow} =
  \begin{pmatrix}
    \rho_\downarrow & \kappa_{\downarrow\downarrow}\\
    \kappa_{\downarrow\downarrow}^\dagger & I-\overline{\rho}_\downarrow
  \end{pmatrix}.
\end{equation}
Here we defined $\rho_\uparrow=\overline{V}_{\uparrow\uparrow}V_{\uparrow\uparrow}^\top, \kappa_{\uparrow\uparrow=}=\overline{V}_{\uparrow\uparrow}U_{\uparrow\uparrow}^\top$ and $\rho_\downarrow=\overline{V}_{\downarrow\downarrow}V_{\downarrow\downarrow}^\top, \kappa_{\downarrow\downarrow=}=\overline{V}_{\downarrow\downarrow}U_{\downarrow\downarrow}^\top$. The total number of electrons is still defined as $N=\Tr{[\rho_\uparrow]}+\Tr{[\rho_\downarrow]}$.

\section{Pole expansion and selected inversion method for solving large-scale HFB calculations}\label{sec:pexsi}

When the number of sites $N_b$ becomes large (e.g. $N_b\gtrsim 10^4$), the direct diagonalization of the quasi-particle Hamiltonian $\mathscr{H}$ can be prohibitively expensive. On the other hand,  if $\mathscr{H}$ is a sparse matrix, we do not need to evaluate all entries of the generalized density matrix $R$, which  is generically a dense matrix. This is because in each step of the SCF iteration, in order to update $\mathscr{H}$, we only need to
evaluate the nonzero entries of $h,\Delta$. Hence when $h_0$ is a sparse matrix and 
$V_{ijkl}$ is a sparse $4$-tensor (such as in the case of the Hubbard-type models), then $h,\Delta$ are also sparse matrices.

Define the sparsity pattern of
$h,\Delta$ as $\mathcal{S}_{h}=\{(p,q) \vert h_{pq}\ne 0\}$,
$\mathcal{S}_{\Delta}=\{(p,q) \vert \Delta_{pq}\ne 0\}$, which can be obtained by evaluating Eq. \eqref{eqn:Hamiltonian_HFB}, \eqref{eqn:HFB_delta}, respectively. Then  $\{\mathcal{S}_{h},\mathcal{S}_{\Delta}\}$ defines the sparsity pattern of $\mathscr{H}$. We  assume that the sparsity pattern of $\mathscr{H}$ does not change along the SCF iteration. Then to solve HFB equations self-consistently,  it is sufficient to evaluate
$\{\rho_{pq} \vert (p,q)\in \mathcal{S}_{h}\}$ and $\{\kappa_{pq} \vert (p,q)\in
\mathcal{S}_{\Delta}\}$ during each iteration. This is guaranteed by Proposition \ref{prop:sparse_hfb}. 

\begin{prop}\label{prop:sparse_hfb}
  Let $R$ be the generalized density matrix of $\mathscr{H}_0$ with sparsity pattern $\{\mathcal{S}_{h},\mathcal{S}_{\Delta}\}$, then generically the evaluation of $\mathscr{H}[R]$ only requires the entries of $R$ restricted to the  same sparsity pattern  $\{\mathcal{S}_{h},\mathcal{S}_{\Delta}\}$.
\end{prop}
\begin{proof}
Assume the statement is not true. Since $h$ only depends only on $\rho$, and $\Delta$ only depends on $\kappa$, respectively, we first consider the evaluation of $h[\rho]$. Then there exists a pair of indices
$(k,l)\notin \mathcal{S}_{h}$ such that $\rho_{kl}$ is needed to evaluate some nonzero entry $h_{ij}$. From Eq.~\eqref{eqn:Hamiltonian_HFB}, this
requires $V_{ikjl}\ne 0$.  From the symmetry property of $V$, we have
$V_{kilj}=V_{ikjl}\ne 0$. Since $\rho$ is generically a dense matrix, the sum $\sum_{ij} V_{kilj}\rho_{ji}$ is generically nonzero, and hence
$h_{kl}$ is generically nonzero. This contradicts the assumption that $(k,l)\notin
\mathcal{S}_{h}$. Similarly for $\kappa$,  if $(k,l)\notin
\mathcal{S}_{\Delta}$ such that $\kappa_{kl}$ is needed in order
to evaluate some nonzero entry $\Delta_{ij}$, then by
Eq.~\eqref{eqn:HFB_delta} we have $V_{ijkl}\ne 0$. This means
$V_{klij}=V_{ijkl}\ne 0$. Then $\sum_{kl} V_{klij} \kappa_{ij}$ and therefore
$\Delta_{kl}$ are generically nonzero, which is a contradiction.  
\end{proof}

Proposition~\ref{prop:sparse_hfb} allows us to use the pole expansion and
selected inversion method (PEXSI) ~\cite{LinLuYingEtAl2009,LinChenYangEtAl2013}
to solve HFB efficiently for large systems. PEXSI is a reduced scaling
algorithm,  its computational cost is at most $\Or(N_b^2)$, and is uniformly applicable to gapped and gapless systems.  Consider a gapped system for simplicity, and we would like to 
choose a contour $\mathcal{C}$ encircling only the negative part of
the spectrum of $\mathscr{H}$.  Then by Cauchy's contour integral formula 
\begin{equation}
  R = \frac{1}{2\pi \I} \oint_{\mathcal{C}} (z - \mathscr{H})^{-1} \ud z.
  \label{}
\end{equation}
After proper discretization of $\mathcal{C}$ using a quadrature, we have
\begin{equation}
  R \approx \sum_{l=1}^{N_\text{poles}} \omega_{l} (z_{l}-\mathscr{H})^{-1} :=
  \sum_{l=1}^{N_\text{poles}} \omega_{l} \mathscr{G}_{l}.
  \label{eqn:R_discretize}
\end{equation}
Here $\mathscr{G}_{l}$ is a quasi-particle Green's function, and can be computed via a matrix inversion. More generally, when finite temperature is under consideration (for gapped and gapless systems), the contour $\mathcal{C}$ can be chosen to be a dumbbell shaped contour encircling the entire spectrum of $\mathscr{H}$, while avoiding the poles of the Fermi-Dirac function. These poles are defined as $x_k=(2k+1)\I\pi/\beta, k\in\ZZ$, which are called the Matsubara poles~\cite{NegeleOrland1988}. After discretizing the contour, the pole expansion still takes the form~\eqref{eqn:R_discretize}.  Fig.~\ref{fig:contour} (a) and (b) illustrate schematically the choice of the contour and its discretization with and without a gap. The discretization points (poles) $\{z_l\}$ can be chosen to be distributed symmetrically with respect to the real axis, which allows us to only evaluate half of the quasi-particle Green's functions.  
This step can be performed using any rational approximation for the Fermi-Dirac function. In particular, the pole expansion in~\cite{LinLuYingE2009}  only scales as $\log(\beta\Delta_E/\epsilon)$, where $\Delta_{E}$ is the spectral radius of $\mathscr{H}$, and $\epsilon$ is the target accuracy. The pre-constant of the logarithmic scaling factor can be further optimized to be numerically near-optimal~\cite{Moussa2016} for approximating Fermi-Dirac functions. 

\begin{figure}[h]
  \begin{center}
    \subfloat[]{\includegraphics[width=0.35\textwidth]{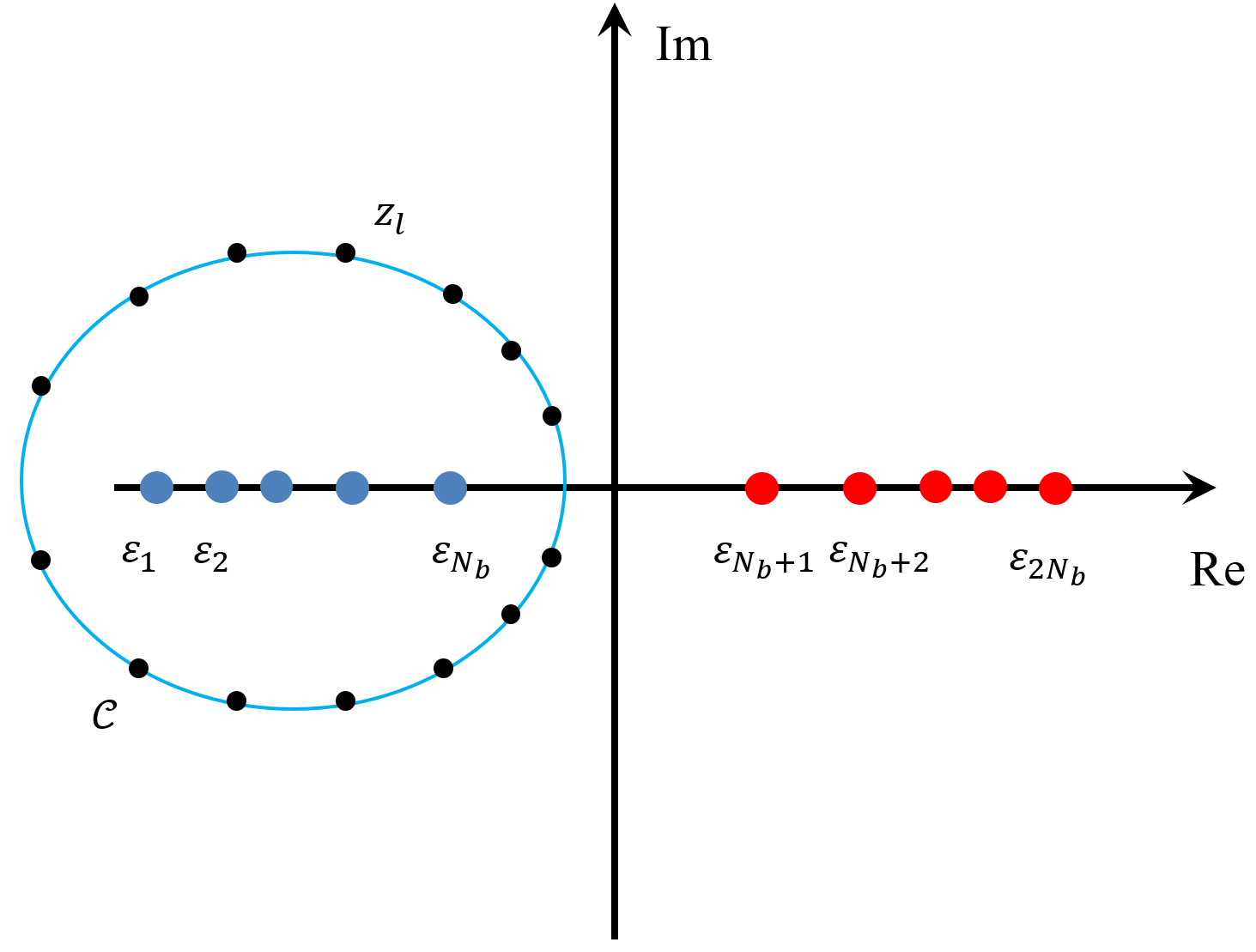}}
    \quad
    \subfloat[]{\includegraphics[width=0.35\textwidth]{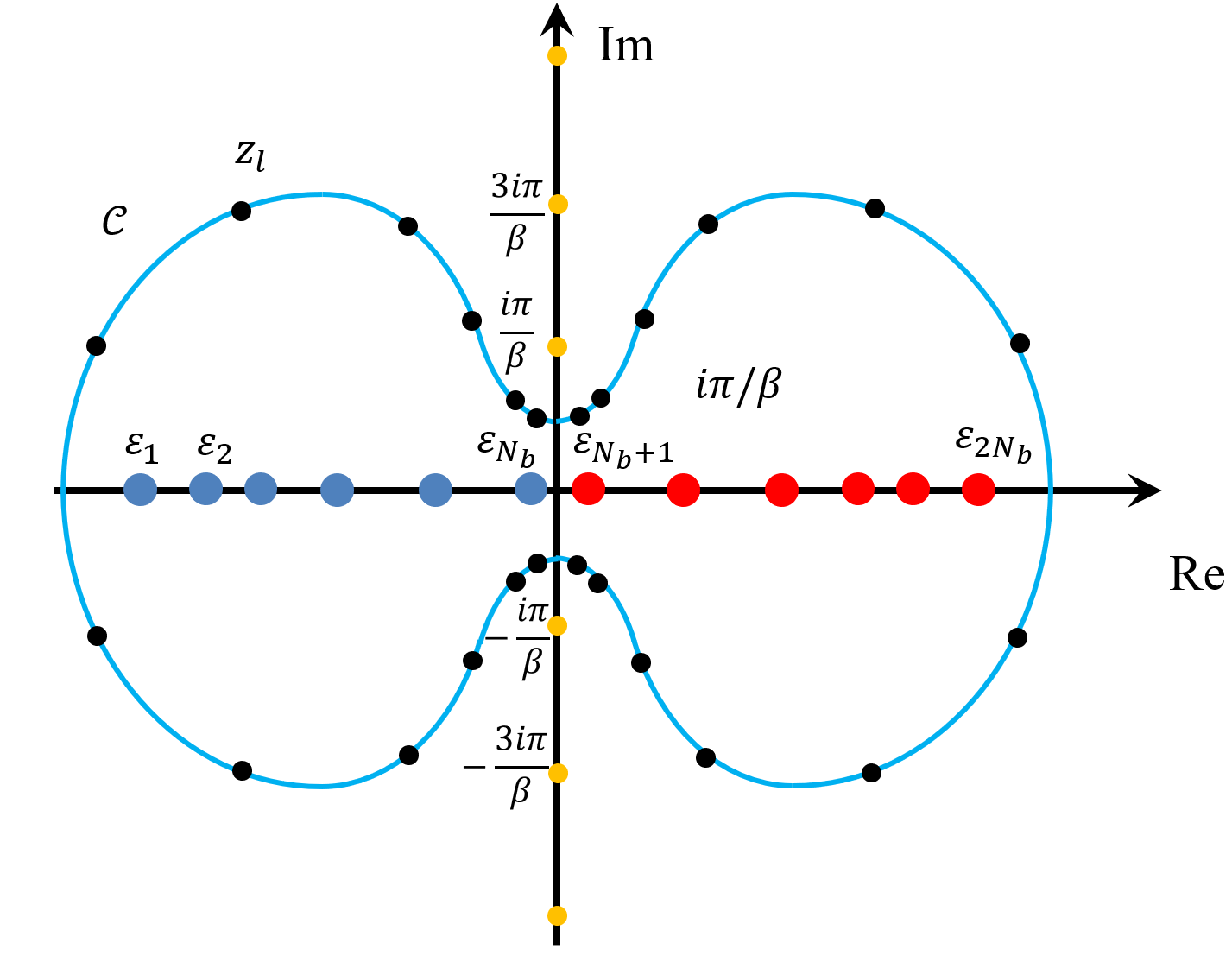}}
  \end{center}
  \caption{Schematic illustration of the contour and its discretization for solving BCS when the quasi-particle Hamiltonian (a) has a relatively large gap; (b) has a small gap or no gap. The yellow points along the imaginary axis indicate the Matsubara poles. The eigenvalues of the quasi-Hamiltonian are always symmetric with respect to the origin.}
  \label{fig:contour}
\end{figure}

Assuming $\{z_l,\omega_l\}$ are given, according to Proposition \ref{prop:sparse_hfb}, we need to evaluate the matrix entries of $R$ corresponding to the sparsity pattern $\{\mathcal{S}_{h},\mathcal{S}_{\Delta}\}$, called the \textit{selected elements} of $R$. This can be performed efficiently using the selected inversion method~\cite{TakahashiFaganChin1973,ErismanTinney1975,LinLuYingEtAl2009,LinYangMezaEtAl2011,KuzminLuisierSchenk2013,JacquelinLinYang2016,JacquelinLinYang2018}, as stated in Proposition \ref{prop:selinv}.

\begin{prop}[Erisman and Tinney~\cite{ErismanTinney1975}]
\label{prop:selinv}
For a matrix $A\in\mathbb{C}^{N\times N}$, let $A=LU$ be its
$LU$ factorization, and $L,U$ are invertible matrices. For any
$1\le k < N$, define
\begin{align}
    \CS_{L} =  \{i\vert L_{i,k}\ne 0\}, \quad \CS_{U} =  \{j\vert U_{k,j} \ne 0\}.
    \label{}
\end{align}
Then all entries $\{A^{-1}_{i,k} \vert i\in \CS_{U}\}$,  
$\{A^{-1}_{k,j} \vert j\in \CS_{L}\}$, and $A^{-1}_{k,k}$ can be
computed using only $\{L_{j,k}| j\in \CS_{L}\}$, $\{U_{k,i}| i\in \CS_U\}$ and $\{A^{-1}_{i,j} \vert (L+U)_{j,i} \ne 0,
i, j\ge k\}$. 
\end{prop}

Therefore in order to evaluate the selected elements of $R$, we only need to evaluate the selected elements of $\mathscr{G}_l$ according to the sparsity pattern of the $LU$ factorization of $\mathscr{H}$. This can lead to significant savings in terms of the computational cost. For a $d$-dimensional lattice system with nearest neighbor interaction (satisfied by Hubbard-type Hamiltonians), the cost for evaluating the selected elements of $R$ scales as $\Or\left(N_{b}^{\max\{3(d-1)/d,1\}}\right)$\cite{LinLuYingEtAl2009,LinYangMezaEtAl2011}. In particular, the cost for $d=2$ and $d=3$ is $\Or(N_b^{1.5})$ and $\Or(N_{b}^2)$, respectively (and we always consider $d\le 3$).

When using the PEXSI method, we should in fact distinguish the cases when $\mathscr{H}$ is real symmetric or Hermitian. This is because in the real symmetric case, $z_l-\mathscr{H}$ is a complex symmetric matrix. Then its $LU$ factorization can be simplified using the complex symmetric $LDL^{\top}$ factorization, and correspondingly the symmetric version of selected inversion can be used\cite{JacquelinLinYang2016}. When $\mathscr{H}$ is Hermitian, the matrix $z_l-\mathscr{H}$ is neither Hermitian nor complex symmetric, but only a structurally symmetric matrix. In such a case, the general $LU$ factorization should be used (e.g. using SuperLU\_DIST\cite{LiDemmel2003}), and correspondingly the selected inversion method for general asymmetric matrices should be used \cite{JacquelinLinYang2018}.  When $\mathscr{H}$ can be block diagonalized, e.g. in the case of spin-singlet and spin-triplet couplings, we may apply PEXSI to each of the sub-matrices of size $N_b\times N_b$, and the overall density matrix and pair matrix can be constructed accordingly.
This makes PEXSI more favorable compared to diagonalization methods for large systems. 
A pseudocode summarizing the solution of HFB equations is described in Algorithm 


\ref{alg:hfb0}.\begin{algorithm}[!htp]
  \caption{Solve a general Hartree-Fock-Bogoliubov problem \label{alg:hfb0}}
  \begin{algorithmic}
    \State \textbf{Input}: $h^0$, $V$, $\rho_0$, $\kappa_0$, $\epsilon$, $\ell=0$
    \State \textbf{Output}: $\rho$, $\kappa$, $E$
    \State Compute  $h,\Delta$ by Eq. \eqref{eqn:Hamiltonian_HFB},\eqref{eqn:HFB_delta}, and form  $\mathscr{H}$. 
    \While {$\delta E < \epsilon$ }
    \While {True}
    \State Compute the generalized density matrix $R$ by diagonalization or  PEXSI.
    \State Update chemical potential $\mu$ in $h$ (using e.g. Newton's method or bisection).
    \If{$|\Tr[\rho_\uparrow]+\Tr[\rho_\downarrow] - N|<\epsilon$}
    \State Break
    \EndIf
    \EndWhile
    \State Compute the total energy $E^{(\ell+1)}$ by \eqref{eqn:HFB_Efunctional}.
    \State Compute  $h,\Delta$ by Eq. \eqref{eqn:Hamiltonian_HFB},\eqref{eqn:HFB_delta}, which gives  $\mathscr{H}[R]$
    \State Update  $\mathscr{H}$ using mixing methods.
    \State Compute the energy difference $\delta E = |E^{(\ell+1)}-E^{(\ell)}|$. 
    \State $\ell\gets \ell+1$ 
    \EndWhile
  \end{algorithmic}
\end{algorithm}



For finite temperature calculations, we may also need to evaluate the Helmholtz free energy at the end of the calculation.  This can be performed using PEXSI as well. First we reformulate the entropy as 
\begin{equation}
  \begin{split}
  S =& \sum_{i} \left\{f_{i} \log (1+e^{\beta \varepsilon_{i}}) + (1-f_{i})\log
  (1+e^{-\beta \varepsilon_{i}})\right\}\\
  = & \sum_{i} \left\{f_{i} \beta \varepsilon_{i} + f_{i} \log (1+e^{-\beta
  \varepsilon_{i}}) + (1-f_{i})\log (1+e^{-\beta \varepsilon_{i}})\right\} \\
  = & \beta \sum_{i} f_{i} \varepsilon_{i} + \sum_{i}\log (1+e^{-\beta \varepsilon_{i}}).
  \end{split}
  \label{}
\end{equation}
Since $E:=\sum_{i} f_{i} E_{i}$ is the total energy at finite temperature, the Helmholtz free energy $\mathcal{F}$ can be expressed as a matrix function of $\mathscr{H}$ as
\begin{equation}
  \mathcal{F} = E - \frac{1}{\beta} S = -\frac{1}{\beta} \sum_{i}\log (1+e^{-\beta
  \varepsilon_i}) = -\frac{1}{\beta} \Tr\log (1+e^{-\beta \mathscr{H}}).
  \label{eqn:helmholtz}
\end{equation}
The function $\log (1+e^{-\beta z})$ is analytic in the complex plane, other than at the Matsubara poles. 
Therefore using the contour integral formulation, 
one can use the same poles as those used for
computing the charge density to evaluate the Helmholtz free energy, but with different weights. We refer readers to ~\cite{LinChenYangEtAl2013,LinGarciaHuhsEtAl2014} for more details.

\section{Numerical experiments}\label{sec:numer}
In this section, we use the 2D Hubbard-Hofstadter model to illustrate the performance of the PEXSI solver for performing large scale HFB calculations. The Hubbard-Hofstadter model extends the classical Hubbard model in the presence of a perpendicular magnetic field. The attractive Hubbard-Hofstadter  model is one of the simplest models exhibiting a nontrivial pairing order, and has a rich ground state phase diagram \cite{wang2014topological, Iskin2015}. The model can be realized in cold atom experiments, which has been used to verify the numerical results \cite{aidelsburger2013realization,kennedy2013spin,cocks2012time}.  

\subsection{Model setup}\label{sec:hhmodel}

Consider a rectangular box containing a lattice of size $N_x\times N_y$. In the Hubbard-Hofstadter model, the non-interacting part of the Hamiltonian is written as
\begin{equation}
  \hat H_0 = -\sum_{i=1}^{N_x}\sum_{j=1}^{N_y}\sum_{\sigma\in\{\uparrow,\downarrow\}}\left[t_x \hat{a}_{i,j,\sigma}^\dagger \hat{a}_{i+1,j,\sigma} + t_y^\sigma(i)\hat{a}_{i,j,\sigma}^\dagger \hat{a}_{i,j+1,\sigma}\right] + \text{h.c.}.
\end{equation}
Here $(i,j)$ is the label of the sites in the  lattice. We set the hopping along the $x$-direction to be $t_x=1$. Here $i+1,j+1$ should be interpreted in the sense of periodic boundary conditions. The hopping along the $y$-direction  depends on both the spatial and spin indices as 
\begin{equation}
  t_y^\sigma (\ell) = e^{\I s_\sigma 2\pi\alpha \ell}.
\end{equation}
Here $s_\uparrow=1, s_\downarrow=-1$. The parameter $\alpha=p/q$ is a rational number, and $p$ and $q$ are coprimes, and $\alpha$ encodes the strength of the magnetic field through the Peierls substitution \cite{LandauLifshitz1991}. The hopping satisfies $t_y^\uparrow(\ell) = \overline{t_y^\downarrow(\ell)}$. The energy spectrum of $\hat H_0$ plotted against $\alpha$ exhibits the celebrated pattern of the Hofstadter butterfly \cite{Hofstadter1976}. 

Here we first consider the a relatively large $\alpha$ ($\alpha=1/3$ and $1/2$) to demonstrate the accuracy of the method, and will consider smaller values of $\alpha$ (i.e. large supercells) later. The interacting part only includes the on-site attractive interaction as
\begin{equation}
  \hat H_1 = -U\sum_{i=1}^{N_y}\sum_{j=1}^{N_y}\hat{n}_{i,j,\uparrow}\hat{n}_{i,j,\downarrow}.
\end{equation}
The total Hamiltonian is then 
\begin{equation}
\hat H = \hat H_0 + \hat H_1.
\end{equation}
Using the matrix notation, the non-interacting part can be rewritten as
\begin{equation}
  \label{eq:16}
  \begin{aligned}
    h_{\uparrow}^0 = t_xI_{N_x}\otimes h_{N_y\times N_y} + h_{N_x\times N_x}\otimes t_y^\uparrow\\
    h_{\downarrow}^0 = t_xI_{N_x}\otimes h_{N_y\times N_y} + h_{N_x\times N_x}\otimes t_y^\downarrow\\
  \end{aligned}
\end{equation}
where
\begin{equation}
  \label{eq:20}
  h_{N\times N} =
  \begin{pmatrix}
    0 & -1 &  & & -1 \\
    -1 & 0 & -1 &\\
     & -1 & \ddots & \ddots \\
     &  & \ddots & 0 & -1\\
    -1 & &  & -1 &0
  \end{pmatrix}_{N\times N}.
\end{equation}
Hence $h_\uparrow^0=\overline{h_\downarrow^0}$.
With some abuse of the notation, we use $i,j$ as general site indices below. The spatial-spin index can be identified with the spin-orbital index according to )
$$p\leftarrow (i,\sigma_1), \quad  q \leftarrow (j,\sigma_2), \quad  r \leftarrow (k,\sigma_3), \quad  s\leftarrow (l, \sigma_4).
$$ 
The interacting term is a highly sparse matrix and can be constructed from antisymmetrizing 
\begin{equation}
  \widetilde{V}_{ijkl, \sigma_1\sigma_2\sigma_3\sigma_4}=
  \left\{
    \begin{array}{cc}
      -U, & i=j=k=l, \sigma_1=\sigma_3=\uparrow, \sigma_2=\sigma_4=\downarrow,\\
      0, & \text{otherwise}
    \end{array}
  \right.
\end{equation}

We consider study the singlet coupling in the Hubbard-Hofstadter model. The HFB energy functional can be simplified as
\begin{equation}
  \label{eq:13}
  \mathcal{E}^\text{HFB}[\rho, \kappa] := \Tr(h_\uparrow^0\rho_\uparrow +h_\downarrow^0\rho_\downarrow) -\mu\Tr(\rho_\uparrow+\rho_\downarrow) - U\sum_{i}\rho_{ii,\uparrow}\rho_{ii,\downarrow} + U\sum_{i}\kappa_{ii,\uparrow\downarrow}\overline{\kappa}_{ii, \uparrow\downarrow}
\end{equation}
According to Eq. \eqref{eqn:Hamiltonian_HFB} and \eqref{eqn:HFB_delta}, the corresponding matrix blocks of the quasi-particle Hamiltonian are expressed as
\begin{equation}
  \label{eq:17}
  h[\rho] =
  \begin{pmatrix}
    h^0_\uparrow-\mu I_{\widetilde N_b}-U\text{diag}(\rho_\downarrow)& 0\\ 
    0 & h^0_\downarrow-\mu I_{\widetilde N_b}-U\text{diag}(\rho_\uparrow)\\ 
  \end{pmatrix}
\end{equation}
and
\begin{equation}
  \Delta[\kappa] = 
  \begin{pmatrix}
    0 & -U\text{diag}(\kappa_{\uparrow\downarrow})\\
    -U\text{diag}(\kappa_{\downarrow\uparrow}) & 0\\
  \end{pmatrix}
\end{equation}
This is a spin-singlet Hamiltonian, and $\mathscr{H}$\ only depends on the diagonals of the density matrix and the pair matrix, respectively. Since $\Delta$ is a block-off-diagonal matrix, $\mathscr{H}$ can be equivalently reduced to the $N_b\times N_b$ system
\begin{equation}
  \mathscr{H}_{\uparrow\downarrow}[\rho, \kappa]
  =\begin{pmatrix} h_{\uparrow}^0 -\mu I_{\widetilde N_b} - U\text{diag}(\rho_\downarrow) &
    -U\text{diag}(\kappa_{\uparrow\downarrow})\\
    -U\text{diag}(\kappa_{\uparrow\downarrow}^{\dagger}) & - \overline{h^{0}_{\downarrow}} + \mu I_{\widetilde N_b} +  U\text{diag}(\rho^{\top}_{\uparrow})
  \end{pmatrix}
\end{equation}
After the solution, we define the following mean pairing potential (with some abuse of notation, denoted by $\overline{\Delta}$ when the context is clear) as the order parameter to characterize the superconducting phase. 
$$
\overline{\Delta}=\frac{U}{\widetilde{N}_b}\sum_{i} \kappa_{ii,\uparrow\downarrow}.
$$

\subsection{Computational details}
Although the periodic boundary condition of the Hubbard-Hofstadter model along the $x$ direction allows us to consider a quasi-1D domain with $N_x\ll N_y$, in order to demonstrate the numerical performance of the method, in this paper we always consider a square lattice with $N_x=N_y=\sqrt{\widetilde{N}_b}$. The ground state under the canonical ensemble (NVT, i.e. the average number of electrons $N$ is fixed) and grand canonical ensemble ($\mu$VT, i.e. the chemical potential $\mu$ is fixed) are investigated separately in the following context. In all calculations, we use a small finite temperature $0.00095$ (corresponding to 300K in the atomic unit). All experiments are performed on the Cori Haswell supercomputer at NERSC, and each node has 32 cores (2 Intel Xeon Processor E5-2698 v3) and 128 GB DDR4 2133 MHz memory.

{\bf Solvers}. We compare the PEXSI solver\footnote{\url{http://www.pexsi.org}} with the  dense eigensolver \texttt{pzheevd} in ScaLAPACK \cite{ScaLAPACK}, and other eigensolvers such as ELPA \cite{Marek_2014} 
can be used as well. SuperLU\_DIST \cite{LiDemmel2003} is used for performing the sparse $LU$ factorization, and ParMETIS \cite{KarypisKumar1998,KarypisKumar1998a} is used to reorder the matrix. The default choice of the number of poles $N_\text{poles}$ is $60$, using the contour integral formulation in \cite{LinLuYingE2009}. For each pole, the nonzero entries of  the quasi-Hamiltonian matrix are distributed across a two-dimensional processor grid of size $N_r\times N_c$. Therefore, PEXSI employs  $N_\text{poles}\times N_{r} \times N_c$ cores in total.  
$N_\text{poles}=60$ is used in most experiments. As a fair comparison, in the dense solver, all the entries of the quasi-Hamiltonian are distributed in  $8N_r\times 8N_c$ processors with the block cyclic data distribution. 

{\bf Self-consistency}. We set the initial density matrix to be zero and the initial pair matrix to be an identity matrix. The self-consistent equation is solved using the  direct inversion of the iterative subspace (DIIS) method \cite{Pulay1980,Pulay1982}, which extrapolates the previous up to 7 steps of the quasi-Hamiltonian. 
  Our convergence criterion is set to be the relative energy difference between two steps (in our calculations $E^{(\ell)}$ always has a non-vanishing amplitude)\ 
\begin{equation}
  \label{eq:24}
  \frac{|E^{(\ell)}-E^{(\ell-1)}|}{|E^{(\ell)}|} < 10^{-9}.
\end{equation}

{\bf Chemical potential}. In the canonical ensemble, the chemical potential is adjusted to meet the requirement of the average electron of numbers in each step of the self-consistent iterations. Consider the function $\nu(\mu) =\Tr[\rho_\uparrow(\mu)] + \Tr[\rho_\downarrow(\mu)]$. Unlike the original HF model where the function $\nu(\mu)$ is a piecewise constant function, this function in the HFB model is most likely a smooth function as shown in Fig. \ref{fig:chemical_potential}.
\begin{figure}[!htb]
  \centering
  \includegraphics[scale=0.4]{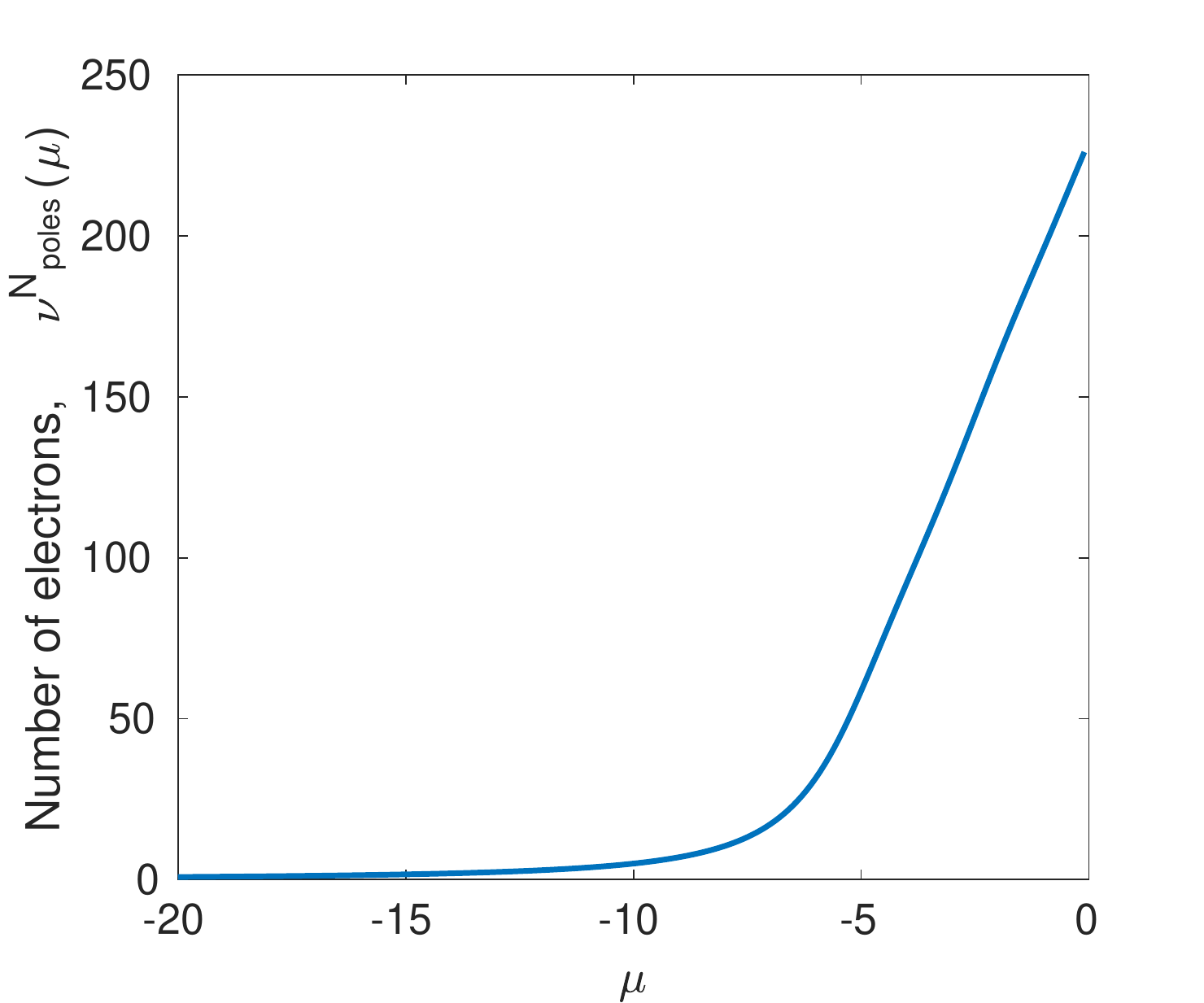}
  \caption{Number of electrons as a function of the chemical potential at self-consistency. $N_x=12$, $N_y=12$, $U=5.0$, $\alpha=1/3$.}
  \label{fig:chemical_potential}
\end{figure}
  
Here we use Newton's method to solve the scalar equation
\begin{equation}
  \label{eq:23}
  \nu(\mu) = N
\end{equation}
within each step of the self-consistent field iteration. The initial guess is set to be zero. Newton's method may fail when $\nu(\mu)$ is not continuously differentiable or the initial guess is not close to the exact solution enough. When the Newton's method fails, the expensive bisection method can be applied  (though we have not observed such cases in our experiments).

To reduce its iteration number, we employed the chemical potential from the last step as the initial guess after the first iteration. The convergence criterion is $(\nu(\mu) - N)/N < 10^{-13}$.  One can also use a looser criterion during the self-consistent loop to save the computational cost.  




\subsection{Accuracy} The overall error of solving HFB equations by PEXSI is from three aspects: 1) approximation of the pole expansion in PEXSI. This is the main error of solving HFB equations caused by the eigensolver PEXSI. The inaccurate generalized density matrix leads to deviation of chemical potential, and further affects the accuracy of total energy. 2) solving chemical potential to meet the particle number by the combination of Newton's method and bisection method. 
3) self-consistent error. Both 2) and 3) are shared between PEXSI and diagonalization methods, and can be systematically reduced to become negligible. The selected inversion is a numerically exact fast algorithm for evaluating selected elements of Green's functions. Hence, we concern about the error due to the pole expansion only.
We first investigate the accuracy of the function $\nu(\mu)$. In each iteration of the self-consistency loop, $\nu(\mu)$ is evaluated at multiple points to adjust the particle number. The accurate evaluation of $\nu(\mu)$ is crucial for solving the HFB equations. From Fig. \ref{fig:nu_mu_err}, the approximation error uniformly decreases as one increases the number of poles.
\begin{figure}
    \centering
    \includegraphics[scale=0.4]{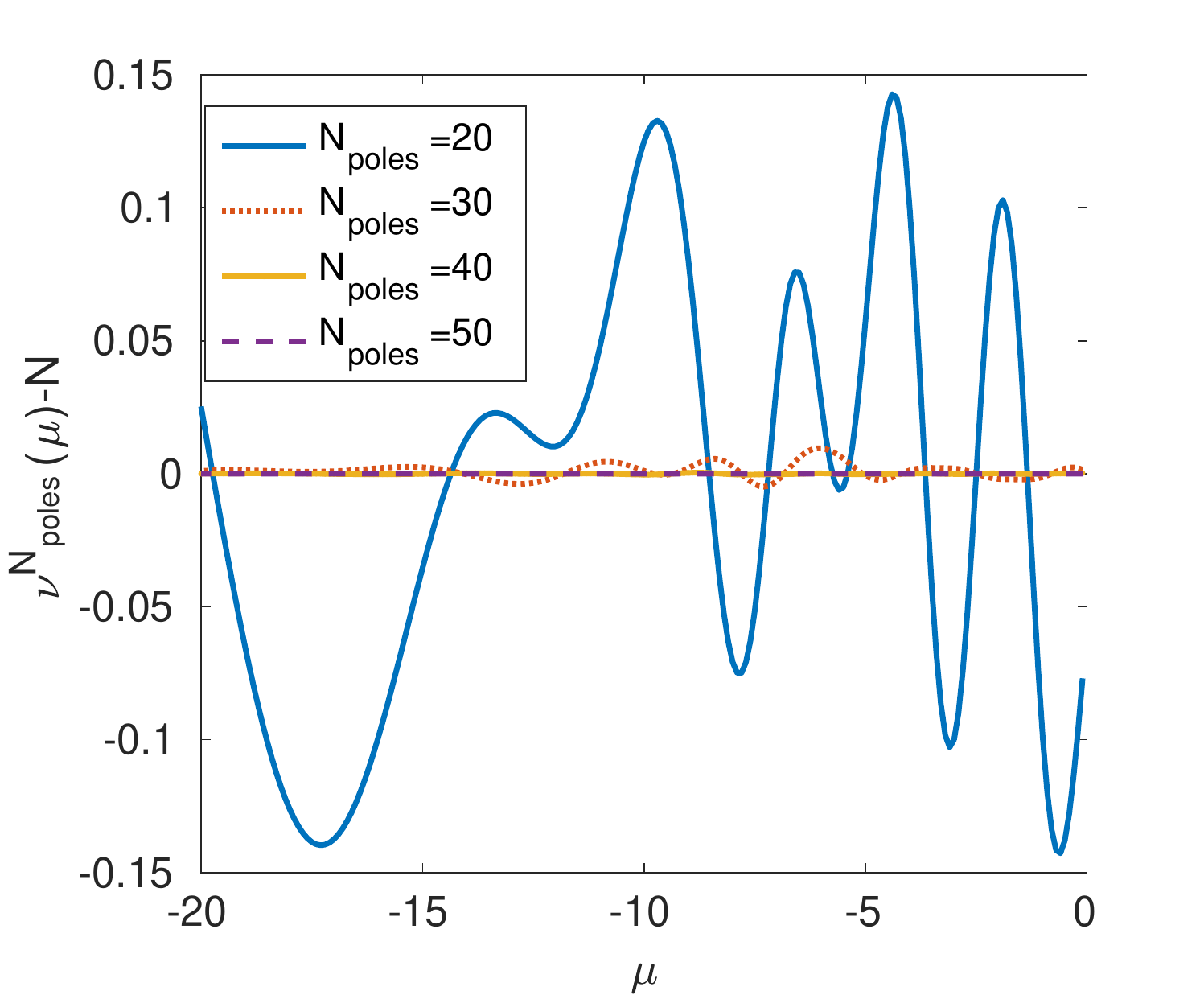}
    \caption{Difference between the function $\nu^\text{exact}(\mu)$ and $\nu^\text{PEXSI}(\mu)$ for $N_\text{poles}$=20, 30, 40, 50. $\nu^\text{exact}(\mu)$: the exact function. $\nu^\text{PEXSI}(\mu)$: the approximated function by PEXSI. $N_x=12$, $N_y=12$, $N=144$, $U=5.0$. 
      The exact values of the function is shown in Fig.\ref{fig:chemical_potential}}
    \label{fig:nu_mu_err}
\end{figure}

We consider the NVT system, and measure the error in terms of the total energy and pairing field. The exact results are obtained by ScaLAPACK. The error of the pairing field is measured in $l_2$ norm. We measure the error with respect to different choices of the number of poles (from $20$ to $80$) and different  system sizes ($24\times 24$ sites to $96\times 96$ sites). With $20$ poles, the energy error and pairing field error are on the order of $10^0$ and $10^{-2}$ , respectively, as shown in Fig. \ref{fig:error}. 
The relative error of energy consistently decreases to $10^{-10}$ as more poles are used for all systems, and the error of pairing field can be reduced to $10^{-9}$. Fig. \ref{fig:error} also shows that the accuracy of PEXSI is not sensitive to the system size. In most calculations below, we set $N_{\text{poles}}=60$.
\begin{figure}[!htbp]
  \centering
  \includegraphics[scale=0.7]{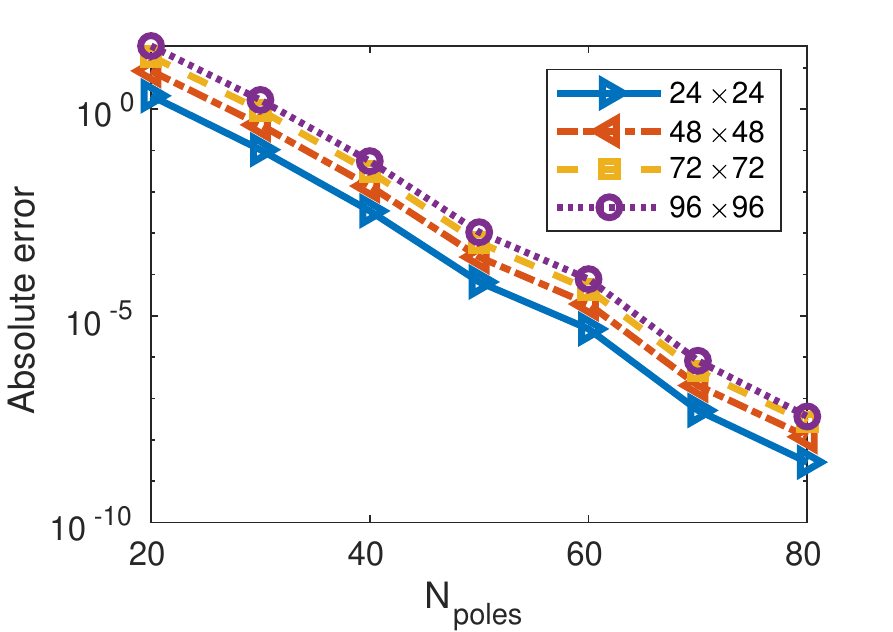}
  \includegraphics[scale=0.7]{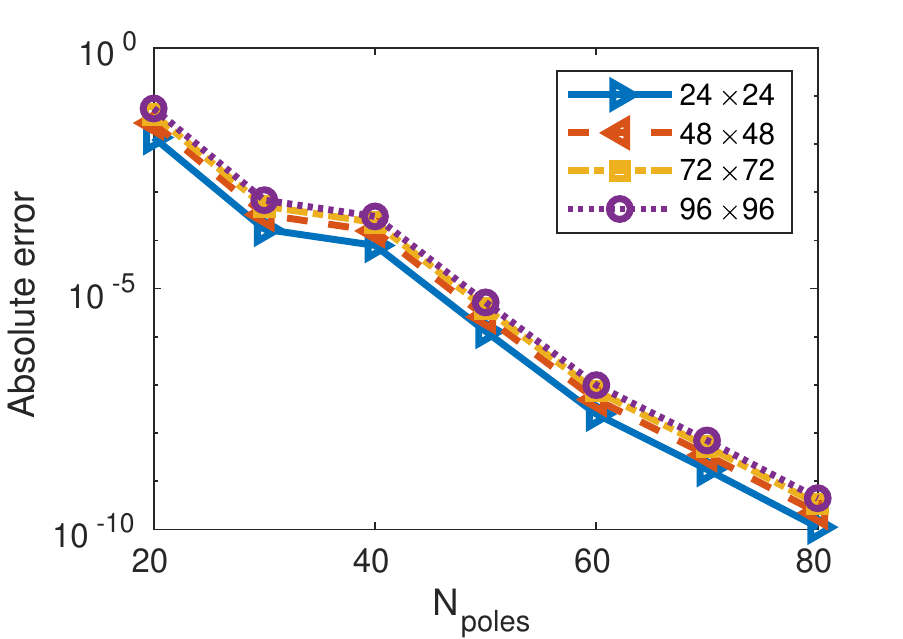}
  \caption{Errors of energy (left) and order parameter (right) by PEXSI with respect to the number of poles. $U=5.0$.}
  \label{fig:error}
\end{figure}

\subsection{Efficiency and scalability} We measure the efficiency of PEXSI in terms of the cost for calculating one step of the self-consistent iteration for $\mu$VT systems.
The efficiency and scalability are investigated by two separated sets of experiments. In the first set, the total number of processors is fixed to be $64$ for ScaLAPACK and $60$ for PEXSI respectively. The exact diagonalization are performed for small systems with $16\times 16$ to $80 \times 80$ sites. For PEXSI, we calculated the systems with $16\times 16$ to $96\times 96$ sites.
The comparison of the wall clock time  of the two methods is presented in Fig. \ref{fig:efficiency}. The cost of exact diagonalization scales asymptotically as $\Or(N^3_{b})$. In our experiments, the scaling is observed to be $\Or(N^{2.94}_{b})$. On the other hand, the asymptotic scaling of the cost of PEXSI is $\Or(N^{1.5}_{b})$ for 2D systems, and the numerically observed scaling is $\Or(N_b^{1.27})$. As shown in Fig. \ref{fig:efficiency}, PEXSI is superior to the exact diagonalization even for a small system. 
The crossover occurs when the system size has around $576$ sites. PEXSI is already about $279$ times faster than the exact diagonalization for a modestly large $80\times 80$ system. 

To demonstrate the strong scaling, we fix the system size as $240\times 240$ for the exact diagonalization, and $1200\times 1200$ for the pole expansion method. The number of processors is raised from $1600$ to $14400$, and from $2304$ to $17280$ for ScaLAPACK and PEXSI, respectively. From Fig.\ref{fig:efficiency}, we observed PEXSI is significantly faster than ScaLAPACK, even though PEXSI is solving for a system that is $25$ times larger. We further quantify the strong scaling by
\begin{equation}
  \label{eq:2}
  \eta(np) = \frac{np0\times WT(np0)}{np\times WT(np)}, \quad  np0 \ll np
\end{equation}
with $WT(np)$ is the wall clock time when $np$ processors are employed. We observed that the parallel efficiency of ScaLAPACK significantly drops down when the number of processors is greater than $10,000$. For example, in ScaLAPACK, $$\eta(14400) \approx \frac{1600\times WT(1600)}{14400\times WT(14400)} = 0.65.$$ While, in PEXSI,$$\eta(17280) \approx \frac{2304\times WT(2304)}{17280\times WT(17280)} = 0.83.$$ Therefore PEXSI has a better parallel efficiency, therefore is more suitable for large scale parallel computation.
\begin{figure}[!htbp]
  \centering
  \includegraphics[scale=0.6]{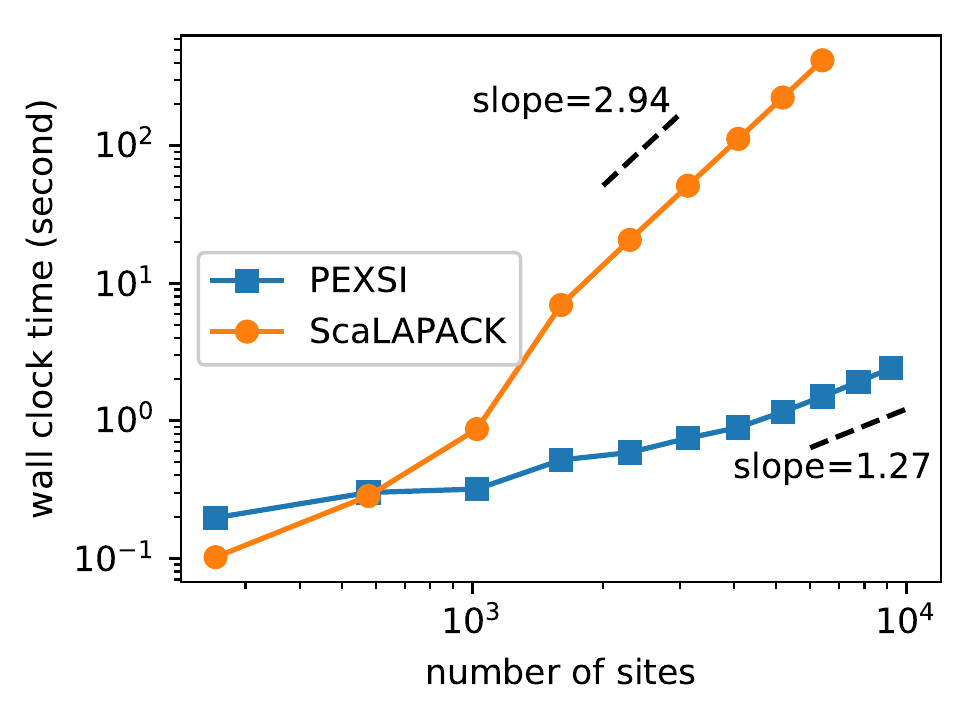}
  \includegraphics[scale=0.6]{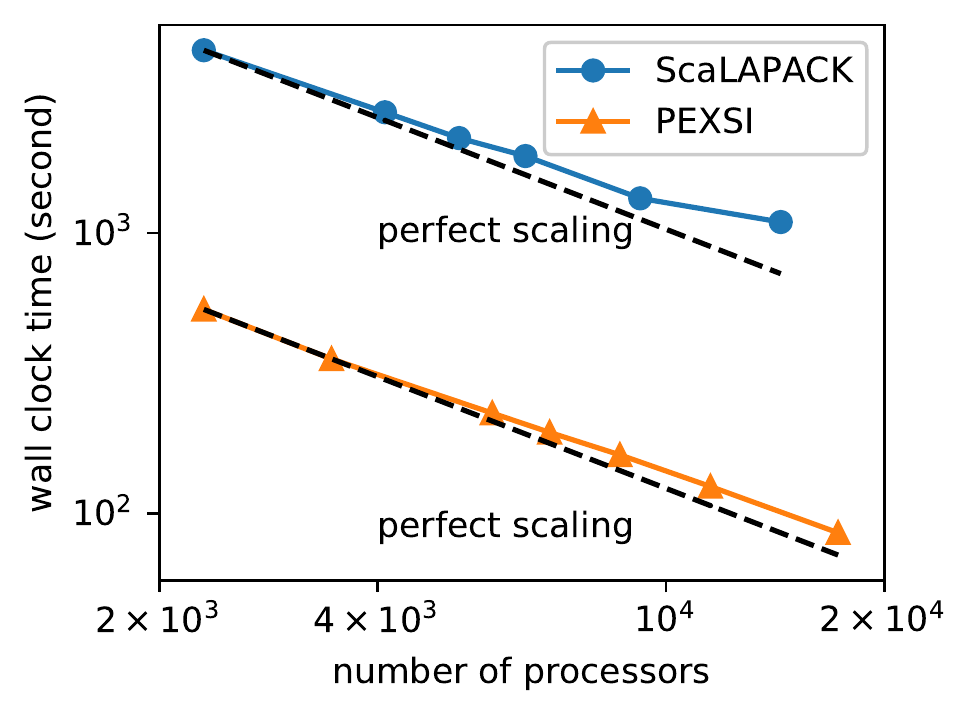}
  \caption{Computational complexity (left) and parallel efficiency (right) of PEXSI and ScaLAPACK for solving HFB equations. Each wall clock time is averaged over $20$ experiments. Left: wall clock time vs number of sites. $\alpha=1/3$, $U=2.0$ and $\mu=-1.0$. $64$ processors are employed by ScaLAPACK, while $60$ processors are employed by PEXSI. The slope for PEXSI is fit using the last $4$ points. The slope of ScaLAPACK is fit using $5$ points. Right: wall clock time vs number of processors. $\alpha=1/3$, $U=2.0$, $\mu=-1.0$. The system size is $240\times 240$ for ScaLAPACK, and $1200\times 1200$ for PEXSI.}
  \label{fig:efficiency}
\end{figure}

\subsection{Phase diagram} We plot the phase diagram of the system with a $180\times 180$ square lattice in terms of the filling factor $n=N/N_b$ and $U$ at $\alpha = 1/3$. 
In terms of the order parameter $\overline{\Delta}$, the mean-field phase diagram of the 2D Hubbard-Hofstadter model has two phases \cite{Iskin2015}: the quantum spin Hall insulator (QSHI) state and the staggered flux (SF) state. We calculated the pairing order parameter of $37 \times 26$ grids in the box with the filling factor $n$ from $2/3$ to $1$, and interaction magnitude $U$ from $1.0$ to $6.0$. The phase diagram is presented in Fig. \ref{fig:phase_diagram}. When $n=1.0$, the system is always in the SF state with a non-vanishing $\overline{\Delta}$. When $n=2/3$, the system can transit from from the QSHI state  to the SF state, and the phase transition occurs at  $U\approx3.0$.  In the cases where $n \in (2/3, 1)$, the order parameter gradually increases with respect to $n$. Our phase diagram result calculated from a finite sized system in real space agrees with that in existing works using both momentum space and real space representations \cite{Zeng2019,Umucalllar2017,Iskin2017}.
\begin{figure}[!htp]
    \centering
    \includegraphics[scale=0.6]{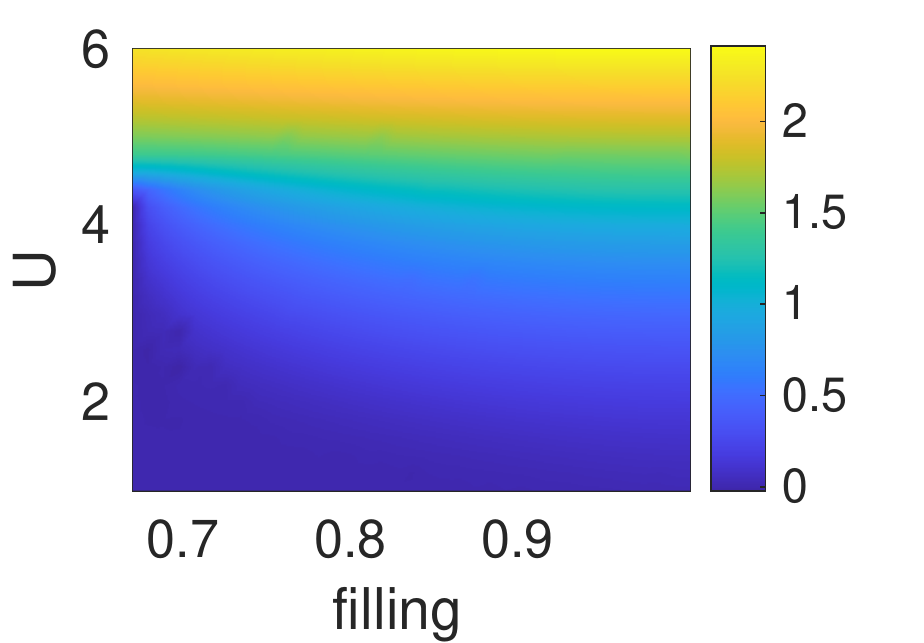}
    \caption{Phase diagram of 2D Hubbard-Hofstadter model. $\alpha=1/3$. $60$ poles. The color bar encodes the order paramter $\overline{\Delta}$.}

    \label{fig:phase_diagram}
\end{figure}

\subsection{Small $\alpha$ factor} The magnitude of the magnetic field $B$ is directly related to the $\alpha$ factor according to the relation
\begin{equation}
  \label{eqn:B_alpha}
  B = \frac{\alpha \Phi_0}{a^2}.
\end{equation}
Here $\Phi_0= 2.0678\times 10^{-15} $ Wb is called superconducting magnetic flux quantum, and $a$ is the lattice spacing. For example, if $\alpha$ is $1/3$ and  $a$ is $4\times 10^{-10}$m, the magnetic field has to be about $4308$ T, which is much larger than the strongest magnetic field generated in any laboratory setting so far (around $1200$ T) \cite{Nakamura2018}.
Therefore in order to simulate the Hubbard-Hofstadter model (and related quantum systems) in the presence of an experimentally achievable magnetic field, we need to be able to simulate with a small $\alpha$ factor.
Since $t_y^\sigma (\cdot)$ should be periodic function, the Hubbard-Hofstadter model with a small $\alpha$ factor need to be simulated with a large lattice.

To demonstrate the capability of the PEXSI accelerated HFB solver, we set $\alpha = 1/480$. Then with the same lattice spacing, the corresponding magnetic field is just about $27$ T. The corresponding lattice has $1440\times 1440\sim 2\times 10^6$ sites. We simulate the $\mu VT$ ensemble with $\mu=-1$ and $U=2.0$. Figure. \ref{fig:small_alpha} shows the spatial distribution of the pairing potential. Six complete pair-density wave (PDW) periods are observed. In each period, there are multiple pair-density waves with shorter wavelength \cite{Iskin2017}. 
\begin{figure}[!htp]
  \centering
  \includegraphics[scale=0.45]{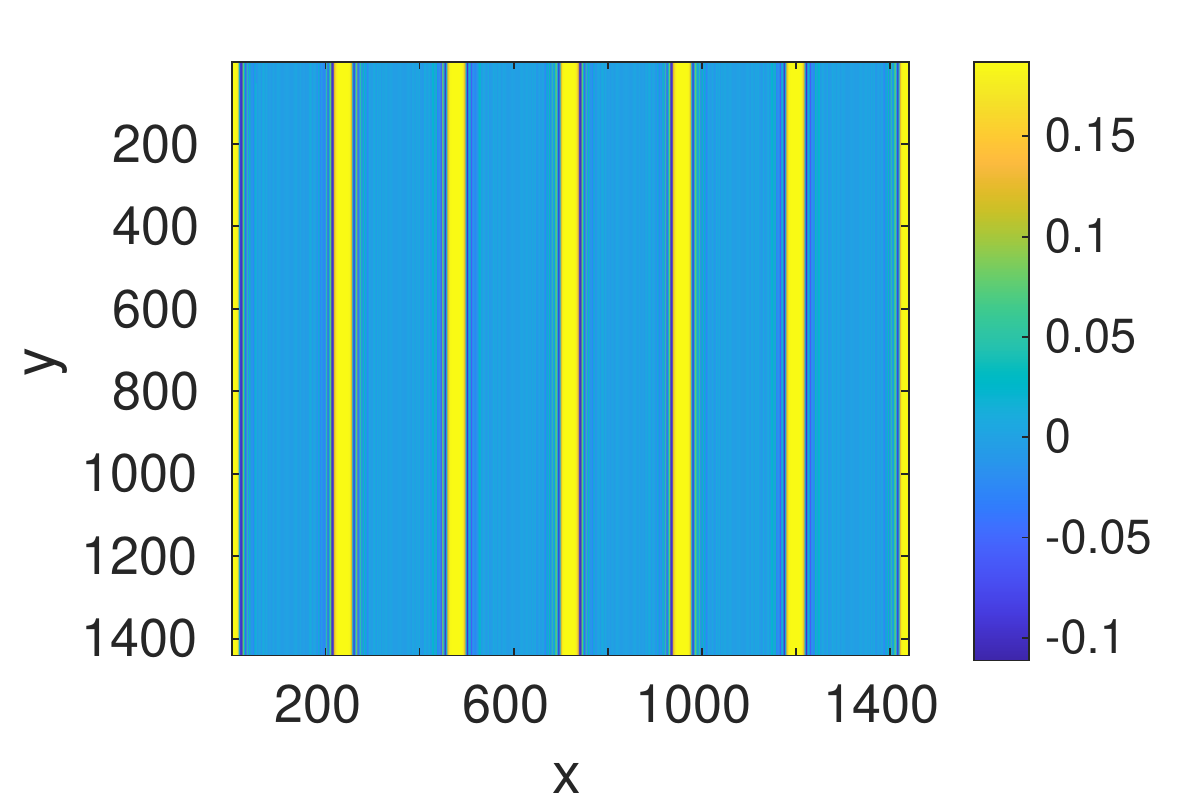}
  \includegraphics[scale=0.45]{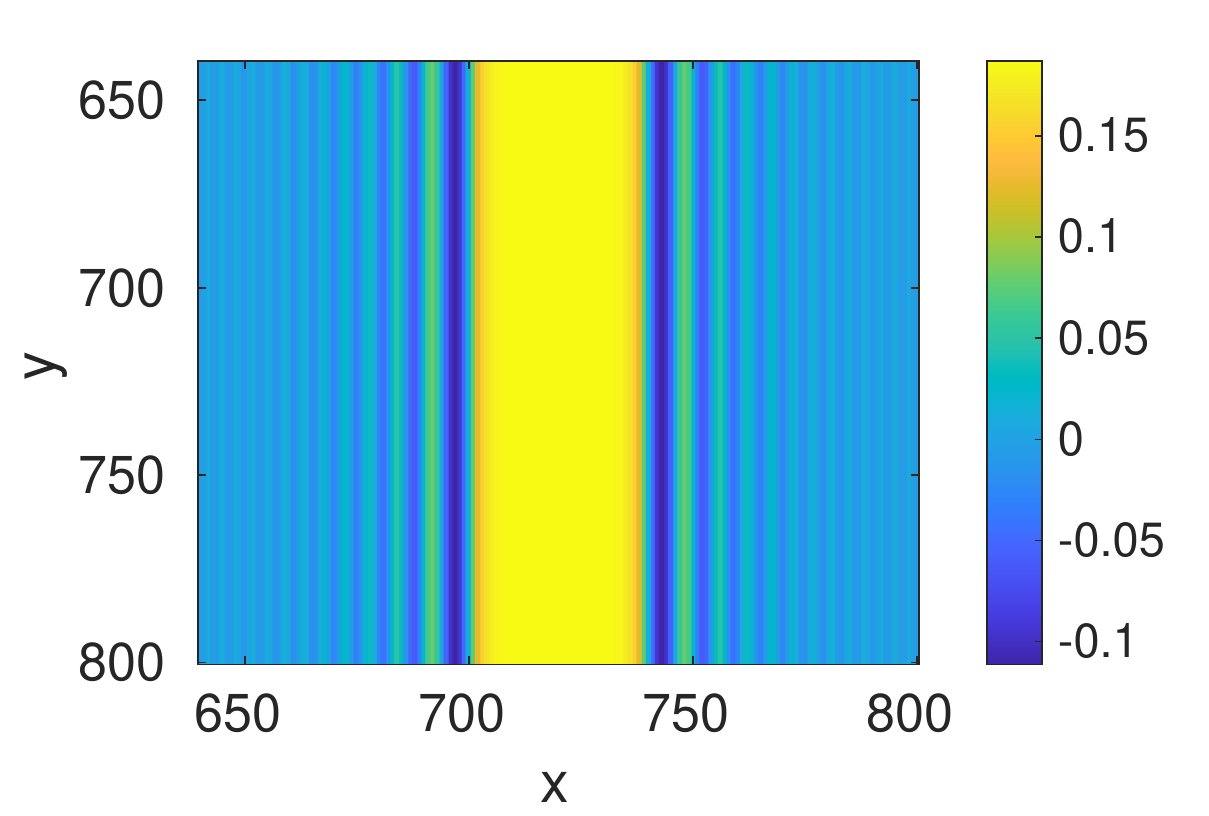}
  \caption{Distribution of pairing potential (left) and its zoom-in (right).}
  \label{fig:small_alpha}
\end{figure}

\section{Conclusion}\label{sec:conclusion}
We present the Hartree-Fock-Bogoliubov (HFB) theory from a numerical perspective, and focuses on the numerical solution of HFB equations for general finite sized quantum systems. When certain spin symmetries (singlet or triplet) are present, the  quasi-particle Hamiltonian, which is a matrix of size $2N_b\times 2N_b$ can be reduced to two matrices of size $N_b\times N_b$. We propose that the pole expansion and selected inversion (PEXSI) method can be well suited for solving large scale HFB equations efficiently, when the quasi-particle Hamiltonian is sparse. The accuracy and efficiency are investigated by the Hubbard-Hofstadter model. We solved a large-scale HFB equations with a relatively weak magnetic field, which attainable in the laboratory setting. Due to the wide use of HFB equations, we expect our solver could be useful for studying a number of physical phenoemena such as the Majorana corner in a large scale Hubbard-Hofstadter model \cite{Zeng2019}. It may also be useful as a mean-field subroutine in quantum embedding methods for treating strongly correlated systems, such as dynamical mean-field theory (DMFT)\cite{GeorgesKotliarKrauthEtAl1996} and density matrix embedding theory (DMET)\cite{KniziaChan2012}.  Such large scale calculations are necessary for understanding e.g. the superconductivity behavior of twisted bilayer graphene (TBG) with magic angles \cite{cao2018unconventional}.

\section*{Acknowledgment}

This work was partially supported by the Air Force Office
of Scientific Research under award number FA9550-18-1-0095,  
by the Department of Energy under Grant No. DE-SC0017867.   
We thank Garnet Chan, Zhihao Cui, Weile Jia, Michael Lindsey, Yu Tong and 
Jianxin Zhu for helpful discussions. We also thank the National Energy Research Scientific Computing Center (NERSC), the Berkeley Research Computing (BRC), and Google Cloud Platform (GCP) for computing resources.

\bibliographystyle{plain}
\bibliography{pexsibcs}

\end{document}